\documentclass[a4paper,11pt]{amsart}

\usepackage[top=2.5cm,bottom=2.5cm, outer=2.2cm, inner=2.2cm]{geometry}
\usepackage{amsfonts,amsmath,amssymb,bbm,amsthm}
\usepackage{enumitem}
\usepackage{stmaryrd}
\usepackage{color,pdflscape} 
\usepackage{hyperref} 

\begin{document}
\newcommand{\R}{\mathbb R}
\newcommand{\Z}{\mathbb Z}
\newcommand{\E}{\mathbb E}
\newcommand{\1}{\mathbbm 1}
\newcommand{\N }{\mathbb N}
\newcommand{\q }{\mathbb Q}
\newcommand{\p }{\mathbb P}
\newcommand{\rr }{\mathsf R}
\newcommand{\D }{\mathcal D}
\newcommand{\F }{\mathcal F}
\newcommand{\G }{\mathcal G}
\newcommand{\h }{\mathcal H}
\newcommand{\M }{\mathcal M}
\newcommand{\eps}{\varepsilon}
\newcommand{\argmax}{\operatornamewithlimits{argmax}}
\newcommand{\argmin}{\operatornamewithlimits{argmin}}
\newcommand{\esssup}{\operatornamewithlimits{esssup}}
\newcommand{\essinf}{\operatornamewithlimits{essinf}}
\newcommand{\sgn}{\operatorname{sgn}}

\newcommand{\nc}{\newcommand}
\nc{\bg}{\begin} \nc{\e}{\end} \nc{\bi}{\begin{itemize}} \nc{\ei}{\end{itemize}} \nc{\be}{\begin{enumerate}} \nc{\ee}{\end{enumerate}} \nc{\bc}{\begin{center}} 
\nc{\ec}{\end{center}} \nc{\fn}{\footnote} \nc{\bs}{\backslash} \nc{\ul}{\underline} \nc{\ol}{\overline} 
\nc{\np}{\newpage}  \nc{\fns}{\footnotesize} 
\nc{\scs}{\scriptsize} \nc{\RA}{\Rightarrow} \nc{\ra}{\rightarrow} \nc{\bfig}{\begin{figure}} \nc{\efig}{\end{figure}} \nc{\can}{\citeasnoun} 
\nc{\vp}{\vspace} \nc{\hp}{\hspace}\nc{\LRA}{\Leftrightarrow}\nc{\LA}{\Leftarrow}
\renewcommand{\tilde}{\widetilde}

 \allowdisplaybreaks

\nc{\ch}{\chapter}
\nc{\s}{\section}
\nc{\subs}{\subsection}
\nc{\subss}{\subsubsection}

\newtheorem{thm}{Theorem}[section]
\newtheorem{cor}[thm]{Corollary}
\newtheorem{lem}[thm]{Lemma}

\theoremstyle{remark}
\newtheorem{rem}[thm]{Remark}

\theoremstyle{example}
\newtheorem{ex}[thm]{Example}

\theoremstyle{ass}
\newtheorem{ass}[thm]{Assumption}

\theoremstyle{definition}
\newtheorem{df}[thm]{Definition}

\newenvironment{rcases}{
  \left.\renewcommand*\lbrace.
  \begin{cases}}
{\end{cases}\right\rbrace}

\setlength\parindent{0pt}
\noindent
\pagestyle{plain}

\title{Second order approximations for limit order books}

%For {amsart}-class:

   \author{Ulrich Horst}
   \address{Humboldt-Universit\"at zu Berlin, Germany}
   \email{horst@math.hu-berlin.de}

   \author{D\"orte Kreher}
   \address{Humboldt-Universit\"at zu Berlin, Germany}
   \email{kreher@math.hu-berlin.de}

	\begin{abstract}

In this paper we derive a second order approximation for an infinite dimensional limit order book model, in which the dynamics of the incoming order flow is allowed to depend on the current market price as well as on a volume indicator (e.g.~the volume standing at the top of the book). We study the fluctuations of the price and volume process relative to their first order approximation given in ODE-PDE form under two different scaling regimes. In the first case we suppose that price changes are really rare, yielding a constant first order approximation for the price. This leads to a measure-valued SDE driven by an infinite dimensional Brownian motion in the second order approximation of the volume process. In the second case we use a slower rescaling rate, which leads to a non-degenerate first order approximation and gives a PDE with random coefficients in the second order approximation for the volume process. Our results can be used to derive confidence intervals for models of optimal portfolio liquidation under market impact.

	\end{abstract}

  \subjclass[2010]{60F17, 91G80}
   \keywords{Functional central limit theorem, second order approximation, high frequency limit, limit order book}

\thanks{Financial support through the CRC TRR 190 is gratefully acknowledged.}

\maketitle

\renewcommand{\baselinestretch}{1.15}\normalsize
\setlength{\parskip}{5pt}

\s{Introduction}

A significant part of financial transactions is nowadays carried out through electronic limit order books. A limit order book is a record, maintained by an exchange or specialist, of unexecuted orders awaiting execution. Incoming limit orders can be placed at many different price levels, while incoming market orders are matched against standing limit orders according to a set of priority rules. The inherent complexity of limit order books renders their mathematical analysis challenging. In this paper we derive a second order approximation for an infinite dimensional limit order book model from individual order arrivals and cancelation dynamics. Depending on the choice of rescaling of the deviations of the actual order book dynamics from the first order approximation derived in \cite{HK1}, we get two second order approximations with different high-frequency dynamics. We illustrate how the different second order approximations can be used to derive confidence intervals for the liquidation values of large portfolios under different forms of market impact. 

Scaling limits for limit order markets have recently gained increasing attention in the probability and mathematical finance literature. Within a Markovian queuing model describing an order book with finitely many price levels diffusive high frequency limits for the price processes were for example obtained in \cite{AbergelJedidi2011,Cont1,Rosenbaum}. Furthermore, a diffusive behaviour for the volumes at the top of the book together with a non-diffusive price process was derived in \cite{Cont3}. In \cite{Lakner2} the authors study a one-sided measure-valued order book model, for which the scaling limit is given by a diffusive price process together with a block-shaped order book for the volumes. A more macroscopic perspective has been adopted in \cite{Marvin} and \cite{ZhengZ} where order book dynamics are described as the solution to SPDEs.

To provide microscopic foundations for PDE or SPDE models one has to consider an infinite number of queues together with a tick size converging to zero in the high frequency limit. 
Depending on the scaling assumptions either fluid limits (cf.~\cite{Gao}, \cite{HK1}, \cite{HP}) or diffusion limits (cf.~\cite{BHQ}, \cite{HK2}) for measure-valued order book models have been derived in the literature. However, to the best of our knowledge no genuine second order approximation for measure-valued limit order book models has been studied so far. Only \cite{Guo} considers a second order approximation in a limit order book setting, albeit for an individual order position and not for the whole limit order book.  

In \cite{HK1} we have studied a class of multiscale Markovian limit order book models, which were characterized by two different time scales: a fast one corresponding to the volume changes and a slow one corresponding to the price changes. Since price changes occur much less frequently than limit order placements in real world markets, this seemed to be a reasonable modelling assumption; see \cite{HP} for empirical evidence. Under certain assumptions on the scaling parameters we have shown that  the discrete limit order book dynamics converges in probability to the solution of a deterministic partial differential equation. The solution can be used to obtain endogenous shape functions for models of optimal portfolio liquidation. In such models the goal is to find optimal strategies for unwinding large numbers of shares over small time periods. They typically assume that the dynamics of standing buy (or sell) side volumes can be described in terms of {\sl exogenous} shape functions. Calibrating the model parameters to market data (see \cite{HP} for a first calibration approach) the previously obtained first order approximation allows for a derivation of {\sl endogenous} shape functions from observable order arrival and cancelation dynamics. Following the same modelling framework this paper derives two second order approximations and illustrates how they can be used to obtain confidence intervals for the liquidation value of large portfolios.

We show that it is {\it not} possible to derive {\it simultaneously} a first order approximation and a second-order approximation which are driven by {\it both}, random fluctuations resulting from price changes {\it and} random fluctuations resulting from the placement and cancelation of limit orders. The intuitive reason is that price and volume dynamics evolve on different time scales. Specifically, our main result shows that  - depending on the choice of rescaling rate - the second order approximation to the limit order book dynamics converges in law to a degenerate SDE-SPDE system. Assuming that the first order approximation of the price process is constant, we can rescale by a fast rescaling rate, under which the price fluctuations can be described by an SDE and the fluctuations of the volumes can be described by an infinite dimensional SDE in the scaling limit. If we allow for a non-degenerate first order approximation, we have to use a slower rescaling rate, under which the fluctuations of the price process converge to an Ornstein-Uhlenbeck process, coupled with a partial differential equation with random coefficients describing the fluctuations of the volume process. 

Our scaling results are intuitive. In high-frequency markets, many placements and cancelations occur inbetween price changes. Hence on the level of event time, prices are constant over short time periods. Over such periods, the fluctuations of the volume density functions are most naturally described by fluctuations resulting from order placements and cancelations. Over longer time periods, price fluctuations are most naturally described by diffusion processes and it is reasonable to assume that the fluctuations of relative (to the best bid and ask prices) volumes result primarily from price fluctuations. Within a portfolio liquidation model we show that the different scaling regimes correspond to different forms of market impact. 

This paper is structured as follows: In the next section we recall the modelling framework of and the weak law of large numbers derived in \cite{HK1}. Section \ref{2nd} contains some preliminary considerations for the derivation of a central limit theorem type of result as well as a statement of our main result on the limiting dynamics of the second order approximation for two different rescaling regimes. The proof of this result can be found in Section \ref{fast} for the fast rescaling regime and in Section \ref{slow} for the slow rescaling regime. In Section \ref{application} we illustrate how our result can be applied to portfolio liquidation problems. The paper concludes with a brief summary and discussion of our findings in Section \ref{end}.

\textsl{Notation.} For two random sequences $(a_n)_{n\in\N},(b_n)_{n\in\N}$, we write $a_n = o_{\mathbb{P}}(b_n)$ if $\lim_{n \to \infty}\frac{a_n}{b_n} = 0$ in probability and $a_n = \mathcal{O}_\p(b_n)$ if $\limsup_{n \to \infty}\left|\frac{a_n}{b_n}\right|$ is bounded in probability. Furthermore, we use the convention $0/0:=0$ throughout the paper.

\s{Setup and first order approximation}\label{setup}

The goal of this paper is to derive a second-order approximation for the sequence of Markovian limit order book models introduced in \cite{HK1}. To this end, we first recall the modelling framework and the first order approximation result for limit order book models established in that paper. For notational convenience we restrict ourselves to one-sided models. Because both sides are basically symmetric, all the results that follow can easily be extended to a two-sided model as defined in Section 5 of \cite{HK1}. Throughout, all random variables are defined on a common complete probability space $(\Omega,\F,\p)$.

\subsection{The model}

The dynamics of the buy side of the limit order book in the $n$-th model is described by a c\`adl\`ag stochastic process $S^{(n)}=\left(S^{(n)}(t)\right)_{0\leq t\leq T}$ taking values in the Hilbert space
\[E:=\R\times L^2(\mathbb{R}),\qquad \left\Vert \alpha\right\Vert_E:=\left|\alpha_1\right|+\left\Vert \alpha_2\right\Vert_{L^2}.\]
The state of the book changes due to arriving market and limit orders and cancelations. In the $n$-th model there are $T_n:=\left\lfloor T/\Delta t^{(n)}\right\rfloor$ such events taking place at times
\[t_k^{(n)}:=k\Delta t^{(n)},\quad k=1,\dots,T_n,\]
where $\Delta t^{(n)}$ denotes a scaling parameter converging to zero as $n\ra\infty$ and $t_0^{(n)}=0$. 
The state of the book after $k$ events is denoted $S^{(n)}_k:=\left(B_k^{(n)},u_k^{(n)}\right)$ and
\[S^{(n)}(t):=\left(B^{(n)}(t),u^{(n)}(t)\right):=\left(B_k^{(n)},u_k^{(n)}\right)\quad\text{for} \quad t\in\left[t_k^{(n)},t_{k+1}^{(n)}\right)\cap[0,T].\]
The real-valued process $B^{(n)}$ describes the best bid price process and the $L^2$-valued process $u^{(n)}$ describes the buy side volume density functions {\it relative} to the best bid price. 

The {\it tick size} is denoted $\Delta x^{(n)}$ and  $x^{(n)}_j:=j\Delta x^{(n)}$ for $j\in\Z,\ n\in\N$. Furthermore, for all $n\in\N$ and $x\in\R$ we define the interval $I^{(n)}(x)$ as
\[I^{(n)}(x):=\left(x_j^{(n)},x_{j+1}^{(n)}\right]\quad\text{for}\quad x_j^{(n)}<x\leq x_{j+1}^{(n)}.\]
 
For every $k=0,\dots,T_n$ the $L^2$-valued random variable $u_k^{(n)}$ is supposed to be a c\`agl\`ad step function on the grid $\left\{x_j^{(n)},\ j\in\Z\right\}$. The standing volume available at time $t_k^{(n)}$ at the relative price level $x_j^{(n)},\ j\in -\N_0$, i.e.~at the absolute price level $B_k^{(n)}+x_j^{(n)}$, is given by 
\begin{equation}\label{volume}
\int_{x_{j-1}^{(n)}}^{x_j^{(n)}}u_k^{(n)}(x)dx=\Delta x^{(n)}u^{(n)}\left(x_j^{(n)}\right).
\end{equation}

At time $t=0$ the state of the limit order book is deterministic for all $n\in\N$ and denoted by 
\[s_0^{(n)}=\left(B_0^{(n)},u^{(n)}_{0}\right) \in \mathbb{R} \times L^2(\R).\]

There are three events that change the state of our order book. The buy side limit order book changes at each time $t_k^{(n)},\ k=1,\dots,T_n$, if:
\be
\item[(A):] a market sell order of size $\Delta x^{(n)} u^{(n)}_{k-1}(0)$, i.e.~equal to the current volume at the best bid queue, arrives. In this case the best bid price decreases by one tick. Hence, the relative volume density function shifts one tick to the right.
\item[(B):] a buy limit order is placed inside the spread one tick above the current best bid price. In this case the best bid price increases by one tick and the relative volume density function shifts one tick to the left.
\item[(C):] a buy limit order placement of size $\Delta v^{(n)}\omega_k^{(n)}$ at the relative price level $\eta_k^{(n)}$ takes place. If $\omega_k^{(n)}<0$, this corresponds to a cancelation of volume.
\ee

The assumption that incoming market orders match precisely against the liquidity at the top of the book follows \cite{BHQ, HK1, HP}. Although the assumption is made primarily for mathematical convenience there is some empirical evidence supporting it. For instance, the authors of \cite{Farmer} found that in their data sample around 85\% 
of the sell market orders which lead to price changes match exactly the size of the volume standing at the best bid price. The effect of a market order that does not lead to a price change is equivalent to a cancelation of standing volume. A market order whose size exceeds the standing volume at the top of the book and that would hence move the price by more than one tick is split by the exchange into a series of consecutively executed smaller orders. The size of each such `child order', except the last, equals the liquidity at the current best bid price. Thus, by definition, a single market order cannot move the price by more than one tick. We acknowledge that this order splitting procedure would most naturally be modelled within a non-Markovian rather than Markovian framework using Hawkes processes as in \cite{Abergel2015, Bacry2014, HW, Rosenbaum_Hawkes, Zheng2014}. 

\begin{rem}
Following \cite{HK1, HP} the relative volume density functions are defined on the whole real line in order to model the arrival of spread placements. The restriction of the function $u^{(n)}(t,\cdot)$ to the interval $(-\infty, 0]$ corresponds to the actual buy side of the order book at time $t$; the restriction to the positive half line specifies the volumes placed into the spread should such events occur next. 
We refer to \cite{HP} for further details on the modelling of spread placements. 
\end{rem}

In the definition of C-events, $\Delta v^{(n)}$ is a scaling parameter that determines the magnitude of an individual limit order placement/cancelation. We will assume that $\Delta v^{(n)}$ for $n \to \infty$. The random variable $\eta_k^{(n)}$ determines the relative price level of a limit order placement/cancelation and takes values in the grid $\left\{x_j^{(n)},\ j\in\Z\right\}$. The random variable $\omega_k^{(n)}$ determines the precise size of the order and by its sign also whether the event is a placement or a cancelation. Thus, a C-event will change the available liquidity at price level $\eta_k^{(n)}$ which is described by the integral of $u_{k-1}^{(n)}$ over the interval $I^{(n)}\left(\eta_k^{(n)}\right)$, cf.~(\ref{volume}), by an amount $\Delta v^{(n)}\omega_k^{(n)}$. Since $u^{(n)}$ is a step function, $u^{(n)}(x)$ will change by $\Delta v^{(n)}\omega_k^{(n)}/\Delta x^{(n)}$ for all $x\in I^{(n)}\left(\eta_k^{(n)}\right)$. In the following we set $\eta_k^{(n)}:=\Delta x^{(n)}\left\lceil \pi_k^{(n)}/\Delta x^{(n)}\right\rceil$, where $\pi_k^{(n)}$ is a real-valued random variable. Then, $I^{(n)}\left(\pi_k^{(n)}\right)=I^{(n)}\left(\eta_k^{(n)}\right)$ and hence C-events can mathematically also be described in terms of $\left(\omega_k^{(n)},\pi_k^{(n)}\right)$ instead of $\left(\omega_k^{(n)},\eta_k^{(n)}\right)$.

Event types are determined by a field of random variables $\left(\phi_k^{(n)}\right)_{k,n\in\N}$ taking values in the set $\left\{A,B,C\right\}$. In particular, it is assumed that only one event happens at a time; this is a natural assumption, given that no two orders are executed simultaneously in real markets. In terms of these random variables we can then introduce the placement operator
\[
	M_{k}^{(n),C}(\cdot):=\1_C\left(\phi_k^{(n)}\right)\frac{\omega_k^{(n)}}{\Delta x^{(n)}} \1_{I^{(n)}\left(\pi_k^{(n)}\right)}(\cdot)
\]
that determines the change in the volume density function due to a C-event. Since the relative volume density functions are defined on the whole real line and since they shift one tick to the right (left) if a market sell order (limit buy order placement in the spread) arrives, the impact of a price change on the order book can conveniently be described in terms of the translation operators $T_-^{(n)}$ and $T_+^{(n)}$, which act on functions $f:\R\ra\R$ according to
\[T_-^{(n)}(f)(\cdot):=f\left(\cdot-\Delta x^{(n)}\right),\qquad T_+^{(n)}(f)(\cdot):=f\left(\cdot+\Delta x^{(n)}\right).\]

The dynamics of the order book described above then translate into the following stochastic difference equations. 

\begin{df}[\cite{HK1}]
For each $n\in\N$ the dynamics of the state process $S^{(n)}=\left(B^{(n)},u^{(n)}\right)$ is given by $S_0^{(n)}:=s_0^{(n)}$ and for $k=1,\dots,T_n$,
\begin{equation*}
\begin{split}
	B_{k}^{(n)} &= B^{(n)}_{k-1}+\Delta x^{(n)}\left(\1_B\left(\phi_k^{(n)}\right)-\1_A\left(\phi_k^{(n)}\right)\right) \\
	u_{k}^{(n)} &= u_{k-1}^{(n)}+\left(T_-^{(n)}-I\right)\left(u_{k-1}^{(n)}\right)\1_A\left(\phi_k^{(n)}\right)+\left(T_+^{(n)}-I\right)\left(u_{k-1}^{(n)}\right)\1_B\left(\phi_k^{(n)}\right)+\Delta v^{(n)} M_{k}^{(n),C}.
\end{split}
\end{equation*}
\end{df}

We make the following assumptions on the initial values and the driving random variables.

\begin{ass}\label{initial} 
The initial relative volume density function $u^{(n)}_{0}$ is a (non-negative) step-function on the grid $\left\{x_j^{(n)},\ j\in\Z\right\}$, which is uniformly bounded by $M$ and has compact support in $[-M,M]$ for all $n\in\N$. Moreover, 
there exists a continuously differentiable function $u_{0}\in L^2$ and $B_0\in\R$ such that
\[\left\Vert u_0^{(n)}-u_0\right\Vert_{ L^2}=\mathcal{O}\left(\Delta x^{(n)}\right)\qquad \text{and}\qquad \left|B^{(n)}_0-B_0\right|=o\left(\Delta t^{(n)}\right)^{1/2}.\]
\end{ass}
We denote $s_0:=\left(B_0,u_{0}\right)\in E$.

\begin{ass}\label{M}
There exists a constant $M>0$ such that for all $n\in\N$ and $k\leq T_n$,
\[\p\left(\left|\pi_k^{(n)}\right|>M\right)=\p\left(\left|\omega_k^{(n)}\right|>M\right)=0.\]
\end{ass}

For every $n\in\N$ and $k=0,1,\dots,T_n$ we define the $\sigma$-field $\F_k^{(n)}:=\sigma\left(S_j^{(n)},\ j\leq k\right)$. We will assume that for each $n\in\N$ the first (and later on also the second, cf.~Assumption \ref{g} below) conditional moments of the state process $S^{(n)}$ depend on the current price and a volume indicator. To define the volume indicator $Y^{(n)}$ we fix throughout a function $h\in L^2(\R)$ with support in $\R_-$ and set
\[Y_k^{(n)}:=\left\langle h, u_k^{(n)}\right\rangle,\ k=0,\dots,T_n,\ n\in\N. \]
The volume indicator can for example model the volume standing at the top of the book. In this way one can account for the empirically well documented fact that volumes at the top of the book (and volume imbalances in a two-sided model) are important determinants of the order flow.

The next assumption describes the Markovian structure of our limit order book model satisfied by the conditional first moments of order placements / cancelations and price movements. 

\begin{ass}\label{Markov}$\quad$
\begin{enumerate}
\item There are Lipschitz continuous functions $p^{(n),A},p^{(n),B}:\R\times\R\ra[0,1]$ with Lipschitz constant $L$ (not depending on $n$) and a scaling parameter $\Delta p^{(n)}$ such that for all $n\in\N$ and $k=1,\dots,T_n$,
\begin{eqnarray*}
\p\left(\left.\phi^{(n)}_k=I \ \right|\F_{k-1}^{(n)}\right)=\Delta p^{(n)} p^{(n),I}\left(B_{k-1}^{(n)},Y_{k-1}^{(n)}\right) \quad a.s.\quad\text{for}\ I=A,B.
\end{eqnarray*}
Moreover, there exist functions $p^A,p^B:\R\times\R\ra[0,1]$ such that
\[\sup_{(b,y)\in\R^2}\left|p^{(n),I}(b,y)-p^I(b,y)\right|=o\left(\Delta x^{(n)}\right)^{1/2},\quad I=A,B.\]
\item There are Lipschitz continuous functions $f^{(n)}:\R\times\R\ra L^2,\ n\in\N,$ with common Lipschitz constant $L>0$ such that for all $k=1,\dots,T_n$,
\[
	\qquad \Delta x^{(n)} f^{(n)}\left(B_{k-1}^{(n)},Y_{k-1}^{(n)};\cdot\right)=\E\left(\left. \1_C\left(\phi_k^{(n)}\right)\omega_k^{(n)} \1_{I^{(n)}\left(\pi_k^{(n)}\right)}(\cdot)\right|\F_{k-1}^{(n)}\right)\quad a.s.	\]
as well as
\[\sup_{(b,y)\in\R_+\times\R}\left\Vert f^{(n)}(b,y;\cdot)\right\Vert_\infty\leq K,\]
for some constant $K<\infty$. Moreover, there exists a function $f:\R\times\R\ra L^2$ such that
\[\sup_{(b,y)\in\R^2}\left\Vert f^{(n)}(b,y)-f(b,y)\right\Vert_{L^2}=\mathcal{O}\left(\Delta x^{(n)}\right),\]
where $f(b,y;\cdot):\R\ra[-M,M]$ is continuously differentiable in $x$ for all $(b,y)\in\R\times\R$ with derivate being uniformly bounded in absolute value by $M$. 
\end{enumerate}
\end{ass}

\subsection{The first order approximation}

Under a suitable scaling assumption $S^{(n)}$ can be approximated by a deterministic process $S=(B,u)$ solving an ODE-PDE system. The approximation requires two time scales: a fast one for the changes in volume and a slow one for the price changes. The main idea behind deriving the first order approximation is that each event should have an expected impact of order $\Delta t^{(n)}$ on the standing volume, since there are $T_n=\lfloor T/\Delta t^{(n)}\rfloor$ order book events.  A price change occurs with probability $\mathcal{O}\left(\Delta p^{(n)}\right)$ by part (1) of Assumpion \ref{Markov} and causes a shift of the volume density function by $\Delta x^{(n)}$. Hence, its expected impact is of order $\Delta p^{(n)}\Delta x^{(n)}$. This requires $\Delta p^{(n)}\Delta x^{(n)} = \mathcal{O}\left(\Delta t^{(n)}\right)$. A placement/cancelation (C-event) changes the volume function by $\mathcal{O}\left(\Delta v^{(n)}/\Delta x^{(n)}\right)$ on the interval of the submission price level. By part (2) of Assumption \ref{Markov} this has an expected impact of $\Delta v^{(n)}$ on the standing volume. This motivates the following scaling assumption.

\begin{ass}\label{scaling}
There exists $\alpha\in(0,1)$ and $\beta\geq 1-\alpha$ such that
\[\Delta t^{(n)}=\Delta v^{(n)},\qquad \Delta x^{(n)}=\left(\Delta t^{(n)}\right)^\alpha,\qquad \Delta p^{(n)}=\left(\Delta t^{(n)}\right)^\beta.\]
\end{ass}

The choice $\beta=1-\alpha$ corresponds to the critical case, which yields a non-trivial price process in the first order approximation. In the subcritical regime $\beta>1-\alpha$, price movements are very rare and hence the first order approximation yields a constant price process. To be precise we have the following weak law of large numbers, which was proven in \cite{HK1}, cf.~Theorem 2.11 and Remark 2.10 of \cite{HK1}.\footnote{In \cite{HK1} the claim was proven for $p^{(n)}\equiv p$, not depending on $n\in\N$. However, the proof in \cite{HK1} can easily be extended to more general $p^{(n)}$ as long as $p^{(n)}$ converges uniformly to some $p$.} We set: $p^{B-A}:=p^B-p^A$.

\begin{thm}\label{LLN}
Under Assumptions \ref{initial}, \ref{M}, \ref{Markov}, and \ref{scaling} there exists a deterministic process $S:[0,T]\ra E$ such that for all $\eps>0$,
\[\lim_{n\ra\infty}\p\left(\sup_{0\leq t\leq T}\left\Vert S^{(n)}(t)-S(t)\right\Vert_E>\eps\right)=0. \]
The function $S=(B,u)$ is the unique solution to the following ODE-PDE system: for all $(t,x)\in[0,T]\times\R$,
\begin{equation*}%\label{system}
\begin{aligned}
B_t&=B_0+\1_{\{\alpha=1-\beta\}}\int_0^tp^{B-A}(B_s,Y_s)ds,\\
u(t,x)&=u_0(x)+\int_0^t f(B_s,Y_s;x)ds+\1_{\{\alpha=1-\beta\}}\int_0^tp^{B-A}(B_s,Y_s)\partial_xu(s,x)ds,\\
Y_t&=\left\langle h,u(t;\cdot)\right\rangle.
\end{aligned}
\end{equation*}
\end{thm}

\s{Towards a second order approximation}\label{2nd}

The goal of this paper is derive a second order approximation for the state process $S^{(n)}=\left(B^{(n)},u^{(n)}\right)$, similar to a central limit theorem. The main difficulty we face in deriving such a result is the fact that the price process and the volume process live on {\it different} time scales. Indeed, the (expected) number of placements and cancelations is of order $T/\Delta t^{(n)}$, while the (expected) number of price changes is of order $T\Delta p^{(n)}/\Delta t^{(n)}$. However, since the relative volume density process depends on the price process (directly through the shifts and indirectly through the conditional probability distribution of new placements and cancelations), we have to use the {\it same} scaling parameter for rescaling both deviation processes, $B^{(n)}-B$ as well as $u^{(n)}-u$. 

In what follows we write $\delta x_k$ for the increment $x_k - x_{k-1}$ for a discrete time stochastic process $(x_k)_{k \in \N}$.

\subs{Preliminaries}

For now we simply introduce a new scaling parameter $\Delta^{(n)}$ and link it later to the other scaling parameters of the model. We define the rescaled discrete fluctuation processes
\[Z^{(n),B}_k:=\frac{B_k^{(n)}-B\left(t_k^{(n)}\right)}{\left(\Delta^{(n)}\right)^{1/2}},\qquad 
Z^{(n),u}_k(\cdot):=\frac{u_k^{(n)}(\cdot)-u\left(t_k^{(n)},\cdot\right)}{\left(\Delta^{(n)}\right)^{1/2}},
\qquad k=0,\dots T_n,\]
as well as the fluctuations of the volume indicator
\[Z^{(n),Y}_k:=\frac{Y_k^{(n)}-Y\left(t_k^{(n)}\right)}{\left(\Delta^{(n)}\right)^{1/2}}=\frac{\left\langle h,u_k^{(n)}-u\left(t_k^{(n)}\right)\right\rangle}{\left(\Delta^{(n)}\right)^{1/2}},\qquad k=0,\dots T_n.\]

Moreover, we define for any $t\in[0,T]$,
\[Z^{(n),B}(t):=\frac{B^{(n)}(t)-B_t}{\left(\Delta^{(n)}\right)^{1/2}},\qquad Z^{(n),u}(t,\cdot):=\frac{u^{(n)}(t,\cdot)-u\left(t,\cdot\right)}{\left(\Delta^{(n)}\right)^{1/2}},\qquad Z^{(n),Y}(t):=\left\langle h, Z^{(n),u}(t)\right\rangle.\]

We have the following useful lemma. 

\begin{lem}\label{discretization}
If $\Delta t^{(n)}=o\left(\Delta^{(n)}\right)^{1/2}$, then 
\[\sup_{k\leq T_n}\sup_{t\in\left[t_k^{(n)},t^{(n)}_{k+1}\right)}\left| Z^{(n),B}(t)-Z^{(n),B}_k\right|\ra0, \qquad
\sup_{k\leq T_n}\sup_{t\in\left[t_k^{(n)},t^{(n)}_{k+1}\right)}\left\Vert Z^{(n),u}(t)-Z^{(n),u}_k\right\Vert_{L^2}\ra0.\]
\end{lem}

\begin{proof}
First note that for every $t\in\left[t_k^{(n)},t^{(n)}_{k+1}\right)$ we have $Z^{(n),u}_k=Z^{(n),u}\left(\lfloor t/\Delta t^{(n)}\rfloor\right)$ as well as \mbox{$Z^{(n),B}_k=Z^{(n),B}\left(\lfloor t/\Delta t^{(n)}\rfloor\right)$.} Now the uniform boundedness conditions in Assumption \ref{Markov} together with Theorem \ref{LLN} imply
\begin{eqnarray*}
&&\left\Vert Z^{(n),u}(t)-Z^{(n),u}\left(\lfloor t/\Delta t^{(n)}\rfloor\right)\right\Vert_{L^2}\\
&&\qquad\leq\left(\Delta^{(n)}\right)^{-1/2}\int_{\Delta t^{(n)}\lfloor t/\Delta t^{(n)}\rfloor}^t \left(\left\Vert f\left(B_
s,Y_s;\cdot\right)\right\Vert_{L^2}+\left\Vert \partial_x u(s,\cdot)\right\Vert_{L^2} \left|p^B(B_s,Y_s)-p^A(B_s,Y_s)\right|\right)ds\\
&&\qquad\leq C\Delta t^{(n)}\left(\Delta^{(n)}\right)^{-1/2}.
\end{eqnarray*}
Similarly,
\begin{eqnarray*}
\left| Z^{(n),B}(t)-Z^{(n),B}\left(\lfloor t/\Delta t^{(n)}\rfloor\right)\right|_{L^2}&\leq&
\left(\Delta^{(n)}\right)^{-1/2}\int_{\Delta t^{(n)}\lfloor t/\Delta t^{(n)}\rfloor}^t \left| p^{B}(B_s,Y_s)-p^{A}(B_s,Y_s)\right|ds\\
&\leq& \Delta t^{(n)}\left(\Delta^{(n)}\right)^{-1/2}.
\end{eqnarray*}
\end{proof}

In the following we denote for $m\in\N$ by $H^m$ the Sobolev space of order $m$ equipped with the usual Sobolev norm. Moreover, $H^0:=L^2$ and $H^{-m}$ denotes the dual of $H^m$. Then
\[\mathcal{E}': =\bigcup_mH^{-m}\supset\dots\supset H^{-1}\supset L^2\supset H^1\supset\dots\supset\bigcap_mH^m =:\mathcal{E}\subset C_\infty(\R).\]
The convergence concept we will use for the second order approximation is weak convergence in the Skorokhod space $D\left([0, T]; \R\times H^{-3}\right)$ equipped with the usual Skorokhod metric. 

In order to derive the second order approximation for the volumes, the following technical lemma will be needed. It states that Theorem \ref{LLN} also implies the convergence of the `discrete derivatives' of $u^{(n)}(t,\cdot)$ to $\partial_x u(t,\cdot)$ in a weak sense.

\begin{lem}\label{diffu}
If $\varphi\in H^3$, then 
\[\sup_{t\leq T}\left|\left\langle \frac{1}{\Delta x^{(n)}}\left(T_+^{(n)}-I\right)u^{(n)}(t),\varphi\right\rangle-\langle \partial_x u(t),\varphi\rangle\right|=o_\p\left(\left\Vert \varphi\right\Vert_{H^2}\right)\]
and
\[\sup_{t\leq T}\left|\left\langle \frac{1}{\Delta x^{(n)}}\left(T_+^{(n)}-I\right)u^{(n)}(t),\frac{1}{\Delta x^{(n)}}\left(I-T_+^{(n)}\right)\varphi\right\rangle+\left\langle \partial_x u(t),\varphi'\right\rangle\right|=o_\p\left(\left\Vert \varphi\right\Vert_{H^3}\right).\]
\end{lem}

\begin{proof} 
First note that, since $L^1_{loc}\supset L^2$ and $C^2\supset H^3$, $\varphi''$ is almost everywhere equal to the absolutely continuous function $x\mapsto\int_0^x\varphi'''(y)dy,\ x\in\R$, where $\varphi'''$ denotes the weak derivative. We have\small
\begin{eqnarray*}
\left\langle \frac{1}{\Delta x^{(n)}}\left(T_+^{(n)}-I\right)u^{(n)}(t)-\partial_x u(t),\varphi\right\rangle&=&\left\langle u^{(n)}(t),\frac{1}{\Delta x^{(n)}}\left(T_-^{(n)}-I\right)\varphi\right\rangle+\left\langle u(t),\varphi'\right\rangle\\
&=&\left\langle u(t)-u^{(n)}(t),\varphi'\right\rangle-\left\langle u^{(n)}(t),\frac{1}{\Delta x^{(n)}}\left(I-T_-^{(n)}\right)\varphi-\varphi'\right\rangle
\end{eqnarray*}\normalsize
and by Theorem \ref{LLN}
\begin{eqnarray*}
\sup_{t\leq T}\left|\left\langle u(t)-u^{(n)}(t),\varphi'\right\rangle\right|\leq \left\Vert \varphi'\right\Vert_{L^2}\sup_{t\leq T}\left\Vert u(t)-u^{(n)}(t)\right\Vert_{L^2}\ra0.
\end{eqnarray*}
Moreover,
\begin{eqnarray*}
\sup_{t\leq T}\left|\left\langle u^{(n)}(t),\frac{1}{\Delta x^{(n)}}\left(I-T_-^{(n)}\right)\varphi-\varphi'\right\rangle\right|
&\leq& \sup_{t\leq T}\frac{1}{\Delta x^{(n)}}\int^0_{-\Delta x^{(n)}}
\left|\left\langle u^{(n)}(t),\varphi'(\cdot+y)-\varphi'(\cdot)\right\rangle\right|dy\\
&\leq& \sup_{t\leq T}\frac{1}{\Delta x^{(n)}}\int_0^{\Delta x^{(n)}}\int_0^y
\left|\left\langle u^{(n)}(t),\varphi''(\cdot+z)\right\rangle\right|dzdy\\
&\leq&\sup_{t\leq T}\left\Vert u^{(n)}(t)\right\Vert_{L^2}\frac{\Delta x^{(n)}}{2}\left\Vert \varphi''\right\Vert_{L^2}\stackrel{\p}{\longrightarrow}0,
\end{eqnarray*}
because the $u^{(n)},\ n\in\N,$ are uniformly bounded in probability. This proves the first claim. Further,\small
\begin{eqnarray*}
&&\left\langle \frac{1}{\Delta x^{(n)}}\left(T_+^{(n)}-I\right)u^{(n)}(t),\frac{1}{\Delta x^{(n)}}\left(I-T_+^{(n)}\right)\varphi\right\rangle
=\left\langle u^{(n)}(t),\frac{1}{\left(\Delta x^{(n)}\right)^2}\left(I-T_-^{(n)}\right)\left(I-T_+^{(n)}\right)\varphi\right\rangle\\
&=&\left\langle u^{(n)}(t),\frac{1}{\left(\Delta x^{(n)}\right)^2}\left(I-T_-^{(n)}\right)\left(I-T_+^{(n)}\right)\varphi+\varphi''\right\rangle-\left\langle u^{(n)}(t)-u(t),\varphi''\right\rangle+\left\langle u'(t),\varphi'\right\rangle
\end{eqnarray*}\normalsize
and \small
\begin{eqnarray*}
\left\Vert \frac{1}{\left(\Delta x^{(n)}\right)^2}\left(I-T_-^{(n)}\right)\left(I-T_+^{(n)}\right)\varphi+\varphi''\right\Vert_{L^2}&\leq&
\frac{1}{\Delta x^{(n)}}\int_0^{\Delta x^{(n)}}\left\Vert\frac{1}{\Delta x^{(n)}}\left(I-T_-^{(n)}\right)\varphi'(\cdot+y)-\varphi''(\cdot)\right\Vert_{L^2}dy\\
&\leq&\frac{1}{\left(\Delta x^{(n)}\right)^2}
\int_0^{\Delta x^{(n)}}\int_{-\Delta x^{(n)}}^0\left\Vert \varphi''(\cdot+y+z)-\varphi''(\cdot)\right\Vert_{L^2}dzdy\\
&=&\frac{1}{\left(\Delta x^{(n)}\right)^2}
\int_0^{\Delta x^{(n)}}\int_{-\Delta x^{(n)}}^0\left\Vert \int_0^{y+z}\varphi'''(\cdot+u)du\right\Vert_{L^2}dzdy\\
&\leq&\Delta x^{(n)}\left\Vert \varphi'''\right\Vert_{L^2}\ra0.
\end{eqnarray*}\normalsize
Thus, the first and the second term converge to zero by the same arguments as above.
\end{proof}

The assumption that $\varphi\in H^3$, which was made in Lemma \ref{diffu}, cannot be relaxed further for the statement of the lemma to be valid. Therefore, we can only expect to prove convergence of the second order approximations in the space $D([0,T];\R\times H^{-3})$. 

For the rest of the paper we will make the following assumption.

\begin{ass}\label{h}
There exists $h\in H^3$ such that $Y^{(n)}:=\langle u^{(n)},h\rangle$ for all $n\in\N$. 
\end{ass}

Since fluctuations in the price and volume indicator directly affect the dynamics of $S$ through the functions $p^A,p^B,$ and $f$, we will also need the following differentiability assumptions.

\begin{ass}\label{pxx} The functions $p^A,p^B$ are twice continuously differentiable in both arguments and for $I=A,B$,
\begin{equation*}
\quad\sup_{b,y}\left\{\left| p^I_b(b,y)\right|+\left| p^I_y(b,y)\right|+\left| p^I_{bb}(b,y)\right|+\left| p^I_{by}(b,y)\right|+\left| p^I_{yy}(b,y)\right|\right\}<\infty.
\end{equation*}
\end{ass}

\begin{ass}\label{fxx}
The function $f$ is twice continuously differentiable in its first two arguments and 
\begin{equation*}
\quad\sup_{b,y}\left\{\left\Vert f(b,y)\right\Vert_{L^2}+\left\Vert f_b(b,y)\right\Vert_{L^2}+\left\Vert f_y(b,y)\right\Vert_{L^2}+\left\Vert f_{bb}(b,y)\right\Vert_{L^2}+\left\Vert f_{by}(b,y)\right\Vert_{L^2}+\left\Vert f_{yy}(b,y)\right\Vert_{L^2}\right\}<\infty.
\end{equation*}
\end{ass}

Before stating the main result of this paper, we give a simple example, for which all assumptions made so far are fulfilled. The example demonstrates that our model is capable of capturing empirically well-documented facts like the dependence of price movements on volumes at the top of the book and the fact that the order flow activity is highest at price levels close to the spread. However, since ours is a Markovian model, it is not capable of capturing self-exiting order flow as it is for example done in \cite{Abergel2015,HW}. 

\begin{ex}\label{example}
We choose $h(x)=-(\lambda x)^3\exp(\lambda x)\1_{(-\infty,0]}(x),\ x\in\R,$ for some $\lambda>0$. Then $h\in H^3$ and due to the exponential tail, $h$ only gives little weight to orders deeper in the book. The larger we choose $\lambda$, the smaller is the range of price levels taken into account by the volume indicator. Thus, $Y^{(n)}=\langle u^{(n)},h\rangle$ approximates the volume standing at the top of the book. 
We take $\phi_k^{(n)},\omega_k^{(n)},\pi_k^{(n)}$ conditionally independent and suppose that
\[\p\left(\left.\omega_k^{(n)}=+1\right|\F_{k-1}^{(n)}\right)=1-\p\left(\left.\omega_k^{(n)}=-1\right|\F_{k-1}^{(n)}\right)=\exp\left(-\left(Y^{(n)}_{k-1}\right)^3\right)\wedge 1,\]
i.e.~high volumes at the top of the book reduce the number of order placements and increase the number of order cancelations. Denoting by $\Phi$ the cumulative normal distribution function, we let
\[\p\left(\left.\phi_k^{(n)}=A\right|\F_{k-1}^{(n)}\right)=\Delta p^{(n)}\frac{B^{(n)}_{k-1}}{1+B^{(n)}_{k-1}}\left(1-\Phi\left(Y^{(n)}_{k-1}\right)\right)=\Delta p^{(n)}-\p\left(\left.\phi_k^{(n)}=B\right|\F_{k-1}^{(n)}\right),\]
i.e.~high volumes at the top of the bid side make a price increase more probable than a price decrease. Moreover, with this choice of conditional event probabilities we make sure that prices stay non-negative. Finally, we suppose that
\[\p\left(\left.\pi_k^{(n)}\in dx\right|\F_{k-1}^{(n)}\right)=C(x-10)^2(x+10)^2\1_{[-10,10]}(x)dx\]
for a normalizing constant $C>0$, i.e.~limit orders (cancelations) are more likely to be placed (occur) close to the best bid price. With the above specifications we have\small
\begin{eqnarray*}
p^{(n),A}(b,y)&=&\frac{b}{1+b}(1-\Phi(y)),\qquad p^{(n),B}(b,y)=1-\frac{b}{1+b}(1-\Phi(y)),\\ 
f^{(n)}(b,y;x)&=&\left(1-\Delta p^{(n)}\right)\left(2e^{-(y\vee0)^3}-1\right)\frac{C}{\Delta x^{(n)}}\sum_{j\in\Z}\1_{\left(x_j^{(n)},x_{j+1}^{(n)}\right]}(x)\int_{x_j^{(n)}}^{x^{(n)}_{j+1}}(z-10)^2(z+10)^2\1_{[-10,10]}(z)dz,
\end{eqnarray*}\normalsize
i.e.~all assumptions required for the first order approximation are satisfied and we have in the limit, if $\alpha=1-\beta$, for all $(t,x)\in[0,T]\times\R$,
\begin{eqnarray*}
dB_t&=&\left(1-\frac{2B_t(1-\Phi(Y_t))}{1+B_t}\right)dt, \qquad Y_t=\langle u(t),h\rangle,\\
u_t(t,x)&=&\left(1-\frac{2B_t(1-\Phi(Y_t))}{1+B_t}\right)u_x(t,x)+\left(2e^{-(Y_t\vee0)^3}-1\right)C(x-10)^2(x+10)^2\1_{[-10,10]}(x).
\end{eqnarray*}
\end{ex}

\subs{Main result}

The two natural candidates for rescaling are either $\sqrt{\Delta t^{(n)}}$ (corresponding to the fast time scale of volume changes) or $\sqrt{\Delta x^{(n)}}$ (corresponding to the slow time scale of price changes). In the next two sections we will present two types of second order approximations corresponding to these two scaling parameters. The first case models a weak dependence of order dynamics on prices, meaning that price movements disppear in the first order approximation, but can be seen in the second order approximation. In the second case, we allow for a more general dependence structure, but have to restrict ourselves to a renormalization by a much slower time scale. The following theorem is the main result of this paper:

\begin{thm}\label{preview}
Let Assumptions \ref{initial}, \ref{M}, \ref{Markov}, \ref{scaling}, \ref{h}, and \ref{fxx} be satisfied and set $\sigma_B:=\left(p^A+p^B\right)^{1/2}$. 
\be
\item[(a)] If $\Delta^{(n)}=\Delta t^{(n)}$, $\alpha>\frac{1}{2}$, and $\beta=2(1-\alpha)$, then under the additional Assumptions \ref{papb} and \ref{g} (see below) there exists a function $\mu$ such that $Z^{(n)}=\left(Z^{(n),B},Z^{(n),u}\right)$ converges weakly in $D([0,T];\R\times H^{-3})$ to $(Z^B,Z^u)$ being the unique solution, starting from $Z_0^B=0$ and $Z^u(0,\cdot)=0$, to the infinite dimensional SDE
\begin{equation}\label{SPDE1}
\begin{split}
dZ^B(t)&=\mu(Y_t)dt+\sigma_B(B_0,Y_t)dW^B_t\\
dZ^u(t)&=f_b(B_0,Y_t) Z^B(t)dt+f_y(B_0,Y_t) \langle Z^u(t),h\rangle dt+\partial_x u(t) dZ^B(t)+dM_t,
\end{split}
\end{equation}
where $W^B$ and $M$ are independent, $W^B$ is a standard Brownian motion and $M$ is an $L^2$-valued Gaussian martingale with covariance depending on $(B,Y)$ (cf.~Theorem \ref{main} below).

\vskip6pt

\item[(b)] If $\Delta^{(n)}=\Delta x^{(n)}$, $\alpha<\frac{1}{2}$, and $\beta=1-\alpha$, then under the additional Assumption
\ref{pxx} we have weak convergence of $Z^{(n)}=\left(Z^{(n),B},Z^{(n),u}\right)$ in $D\left([0,T];\R\times H^{-3}\right)$ to $\left(Z^B,Z^u\right)$ being the unique weak\footnote{By `weak' solution we mean `weak' in the PDE sense, i.e.~we consider $Z^u$ as a distribution valued process, being an element of $H^{-3}$. The exact solution concept will be explicitly defined in Theorem \ref{main2} below. } solution, starting from $Z^B_0=0$ and $Z^u(0,\cdot)=0$, to the system
\begin{equation}\label{SPDE2}
\begin{split}
dZ^B(t)&=p_b(B_t,Y_t)Z^B(t)dt+p_y(B_t,Y_t)\langle Z^u(t),h\rangle dt+\sigma_B(B_t,Y_t)dW_t,\\
\qquad dZ^u(t)&=f_b(B_t,Y_t) Z^B(t)dt +f_y(B_t,Y_t) \langle Z^u(t),h\rangle dt+\partial_x u(t)dZ^B(t)+ \partial_x Z^u(t)dB_t,
\end{split}
\end{equation}
where $W$ is a standard Brownian motion.
\ee
\end{thm}

\begin{rem}
Assumption \ref{fxx} only requires the smooth differentiability of $f$ in $(b,y)$. It can be seen from the proof of Theorem \ref{tightness}, especially equation (\ref{ZBY}) that under the additional assumption that $(b,y,x)\mapsto f(b,y;x)$ is twice continuously differentiable in all three arguments with uniformly bounded derivatives, $Z^u$ will be pathwise differentiable in $x$ and (\ref{SPDE2}) will have a strong solution.
\end{rem}

\begin{ex}
(Continuation of Example \ref{example})
Part (b) of Theorem \ref{preview} applies to Example \ref{example}. In this case, the second order approximation satisfies
\begin{eqnarray*}
dZ^B(t)&=&\frac{2(1-\Phi(Y_t))}{1+B_t}Z^B(t)dt+\frac{2B_t\Phi'(Y_t)}{1+B_t}\langle Z^u(t),h\rangle dt+dW_t,\\
dZ^u(t,x)&=&\partial_xu(t,x)dZ^B(t)+\partial_xZ^u(t,x)dB_t\\
&&-\1_{[0,\infty)}(Y_t)6Y_t^2e^{-Y_t^3}C(x-10)^2(x+10)^2\1_{[-10,10]}(x)\langle Z^u(t),h\rangle dt.
\end{eqnarray*}
\end{ex}

We will give an example illustrating part (a) of Theorem \ref{preview} in Section \ref{fast}.

\s{Application: risk management for portfolio liquidation}\label{application}

The first order approximation, Theorem \ref{LLN}, can be used to obtain endogenous order book shape functions for models of optimal portfolio liquidation under market impact, cf.~Section 4 in \cite{HP}. In this section we illustrate how the second order approximations can potentially be used to construct confidence intervals for the resulting liquidation values for the two benchmark cases of  strictly permanent and strictly non-permanent price impact.

\subs{Non-permanent price impact}\label{nonpermanent}

In the case of non-permanent price impact, it is supposed that the order book recovers infinitely fast from liquidity shocks. This assumption is appropriate for very liquid stocks. Liquid stocks typically are characterized by high trading rates, small spreads and low price volatility over shorter periods of time. Hence, we may assume that the first order approximation of the price process is constant, i.e.~$B_t\equiv B_0$ for all $t\in[0,T]$, in which case the order book shape function is given by
\[u(t,x)=u_0(x)+\int_0^tf(B_0,Y_s;x)ds\quad\forall\ (t,x)\in[0,T]\times\R.\]

\begin{rem}
The assumption of a constant price process and deterministic shape function is consistent with many liquidation models where trading costs are benchmarked against some fundamental price process that follows a martingale, but where the stochasticity eventually drops out of the optimization problem; cf.~\cite{Alfonsi} and references therein.  
\end{rem}

Let us now consider a large trader who needs to liquidate a single stock portfolio of $X>0$ shares until time $T$ and only trades at times $t_i^{(n)},\ i=0,\dots,T_n$. If $\theta_i^{(n)}$ denotes the number of shares liquidated at time $t_i^{(n)}$, then a first order approximation of the total income from liquidation under strictly non-permanent price impact is given by
\[V\left(\theta^{(n)}\right):=\sum_{i=0}^{T_n}\int_{-c_i^{(n)}}^0\left(B_0+x\right)u\left(t_i^{(n)},x\right)dx=B_0X+\sum_{i=0}^{T_n}\int_{-c_i^{(n)}}^0xu\left(t_i^{(n)},x\right)dx,\]
where 
the $c_i^{(n)},\ i=0,\dots,n,$ are defined by
\begin{equation}\label{c}
\theta_i^{(n)}=\int_{-c_i^{(n)}}^0u\left(t_i^{(n)},x\right)dx\quad \text{ and }\quad X=\sum_{i=0}^{T_n}\theta_i^{(n)}.
\end{equation}

Let $\theta^{(n),*}$ be an optimal strategy for the above liquidation problem. Then $\theta^{(n),*}$ is a deterministic function. 
Since the first order approximation of the price process is supposed to be constant, part (a) of Theorem \ref{preview} gives for large $n$ the approximation
\[B^{(n)}(t)\approx B_0+\sqrt{\Delta t^{(n)}}Z^B(t),\quad u^{(n)}(t,x)\approx u(t,x)+\sqrt{\Delta t^{(n)}}Z^u(t,x),\quad (t,x)\in[0,T]\times\R.\]
The actual liquidation value corresponding to the strategy $\theta^{(n),*}$ can thus be approximated by
\[\tilde{V}\left(\theta^{(n),*}\right):=B_0X+\sum_{i=0}^{T_n}\int_{-d^{(n),*}_i}^0\left(\sqrt{\Delta t^{(n)}}Z^B\left(t_i^{(n)}\right)+x\right)\left(u\left(t_i^{(n)},x\right)+\sqrt{\Delta t^{(n)}}Z^u\left(t_i^{(n)},x\right)\right)dx,\]
where the $d^{(n),*}_i,\ i=0,\dots,T_n,$ are defined via 
\begin{equation} \label{theta*}
	\theta^{(n),*}_i=\int_{-d^{(n),*}_i}^0u\left(t_i^{(n)},x\right)+\sqrt{\Delta t^{(n)}}Z^u\left(t_i^{(n)},x\right)dx.
\end{equation}
In the absence of permanent market impact it is reasonable to assume that the optimal liquidation strategy is such that $\sup_{i=1,\dots,T_n}\theta^{(n),*}_i=o(1)$. 
In this case, it follows from (\ref{c}) and (\ref{theta*}) 
that
\begin{equation} \label{approx1}
\tilde{V}\left(\theta^{(n),*}\right)\approx
V\left(\theta^{(n),*}\right)+\sqrt{\Delta t^{(n)}}\sum_{i=1}^{T_n}\left(Z^B\left(t_i^{(n)}\right)\theta^{(n),*}_i+\int_{-c^{(n),*}_i}^0x Z^u\left(t_i^{(n)},x\right)dx\right),
\end{equation}
where the $c^{(n),*}_i,\ i=0,\dots,T_n,$ were defined in (\ref{c}).
Since the dynamics of $(Z^B,Z^u)$ are in principal known from part (a) of Theorem \ref{preview}  the above representation allows to construct confidence intervals for the liquidation value.

\subs{Permanent price impact}

Let us now consider the opposite case, where price impact is persistent and the price as well as the volume function jump at each trading time of the large trader and then continue to evolve according to the Markovian dynamics under consideration. In this case the first order approximation gives an idea of how resilient the order book actually is, i.e.~how it reacts to state changes. 

If the large trader follows a deterministic discrete liquidation strategy $\theta^{(n)}=\left(\theta_i^{(n)}\right)$ trading at times $t_i^{(n)},\ i=0,\dots,T_n$, then the first order approximation is influenced by the trader's action and the new dynamics are given by the c\`agl\`ad process $(B^\theta, u^\theta)$ defined as 
\begin{align*}
 u^\theta(0,x)&=u_0(x),\\
u^\theta\left(t^{(n)}_i+,x\right)&=u^\theta\left(t_i^{(n)},x-c_i^{(n)}\right)\quad\text{with}\quad\theta_i^{(n)}=\int_{-c_i^{(n)}}^0u^\theta\left(t_i^{(n)},x\right)dx,\quad i=0,\dots,T_n,\\
u^\theta(t,x)&=u^\theta\left(t_i^{(n)}+,x\right)+\int_{t_i^{(n)}}^tp^{B-A}\left(B^\theta_s,Y^\theta_s\right)u_x^\theta(s,x)+f\left(B^\theta_s,Y^\theta_s;x\right)ds,\quad t\in\left(t_i^{(n)},t_{i+1}^{(n)}\right],\\
Y_t^\theta&=\left\langle u^\theta(t),h\right\rangle,\quad t\in[0,T], 
\end{align*}
and similarly
\begin{align*}
B^\theta(0)&=B_0,\qquad  B^\theta\left(t_i^{(n)}+\right)=B^\theta\left(t_i^{(n)}\right)-c_i^{(n)},\quad i=0,\dots,T_n,\\
B^\theta(t)&=B^\theta\left(t_i^{(n)}+\right)+\int_{t_i^{(n)}}^tp^{B-A}\left(B^\theta_s,Y^\theta_s\right)ds,\quad t\in\left(t_i^{(n)},t_{i+1}^{(n)}\right].
\end{align*}
Hence, the first order approximation of the liquidation value of any such liquidation strategy equals 
\[V\left(\theta^{(n)}\right):=\sum_{i=0}^{T_n}\int_{-c_i^{(n)}}^0\left(B^\theta\left(t_i^{(n)}\right)+x\right)u^\theta\left(t_i^{(n)},x\right)dx.\]
Given an optimal strategy $\theta^{(n),*}$, one would like to use the second order approximation from part (b) of Theorem \ref{preview} to construct confidence intervals for the liquidation value $V\left(\theta^{(n),*}\right)$ in a similar way as it was done in Subsection \ref{nonpermanent}. Although the exact construction is more involved in this setting, we expect the following approximation to be reasonable if the initial position is split into sufficiently many small orders $\theta^{(n),*}_i,\ i=1,\dots,T_n$ (and possibly a block trade at time $t_0=0$):
\begin{equation} \label{approx2}
\tilde{V}\left(\theta^{(n),*}\right)\approx
V\left(\theta^{(n),*}\right)+\sqrt{\Delta x^{(n)}}\sum_{i=1}^{T_n}\left(Z^{B,\theta^{*}}\left(t_i^{(n)}\right)\theta^{(n),*}_i+\int_{-c_i^{(n),*}}^0x Z^{u,\theta^{*}}\left(t_i^{(n)},x\right)dx\right),
\end{equation}
where $Z^{B,\theta^*}$ and $Z^{u,\theta^*}$ are defined in a similar fashion as $B^\theta$ and $u^\theta$ above, with their piecewise diffusive dynamics being described in part (b) of Theorem \ref{preview}. 

Comparing the two approximations (\ref{approx1}) and (\ref{approx2}) we obtain the very intuitive result that the fluctuations in the liquidation value are much smaller in the absence of permanent price impact.

\s{Renormalization on a fast time scale}\label{fast}

In this section we will derive a second order approximation to the discrete limit order book model under fast rescaling. In order to avoid that the price fluctuations explode in this case we have to assume that price movements are very rare. 

\begin{ass} \label{scaling1} 
\[\alpha\in\left(\frac{1}{2},1\right),\qquad \beta=2(1-\alpha),\qquad \Delta^{(n)}=\Delta t^{(n)}.\]
\end{ass}

In this case the first order approximation from Theorem \ref{LLN} takes the special form
\begin{eqnarray*}
B_t= B_0,\qquad u(t,x)=u_0(x)+\int_0^tf(B_s,Y_s;x)ds,\qquad Y_t=\langle h,u(t)\rangle\qquad\forall\ (t,x)\in[0,T]\times\R.
\end{eqnarray*}
Especially, the price process $B=(B_t)_{t\in[0,T]}$ is constant under Assumption \ref{scaling1}.

\begin{rem}
Assumption \ref{scaling1} consists of three parts: First, the fast rescaling rate $\Delta^{(n)}=\Delta t^{(n)}$ will ensure that the order placement and cancelation activity will lead to a diffusive behaviour of the volume fluctuations in the limit. Second, in order to control the second moments of the price fluctuations we need that $\beta\geq2(1-\alpha)$. We are here looking at the critical case $\beta=2(1-\alpha)$. If $\beta$ was strictly larger than $2(1-\alpha)$, then in the scaling limit there would be no noise term and hence the price fluctuations would be constant. Together with the third condition, $\alpha>1/2$, which controls the higher moments of the price fluctuations, this will yield the diffusive behaviour of the price fluctuations in the limit. 
\end{rem}

\subs{Fluctuations of the price process}

We define for all $n\in\N$ the function
\[p^{(n)}(b,y):=\left(\Delta t^{(n)}\right)^{1/2-\alpha}\left(p^{(n),A}(b,y)-p^{(n),B}(b,y)\right).\]

\begin{ass}\label{papb}
There is a Lipschitz continuous function $p:\R\times\R\ra\R$ such that
\[\sup_{(b,y)\in\R_+\times\R}\left| p^{(n)}(b,y)-p(b,y)\right|\ra0.\]
\end{ass} 

\begin{rem}
The difference between $p^{(n),A}$ and $p^{(n),B}$ has to scale in $\left(\Delta t^{(n)}\right)^\gamma$ with $\gamma\geq \alpha-1/2$ in order to avoid explosions of the drift part of the price fluctuations. As above we are here considering the critical case corresponding to $\gamma=\alpha-1/2$. Note however that Assumption \ref{papb} implicitely also includes the subcritical case $\gamma>\alpha-1/2$, in which case $p\equiv0$ in the results that follow below. Moreover, Assumption \ref{papb} implies that $p^A=p^B$. 
\end{rem}

Since $B^{(n)}$ converges to the constant process $B_t=B_0,\ t\in[0,T]$, we define the drift and volatility of the price fluctuation process in the limit as
\[\mu(y):=p(B_0,y),\qquad \sigma(y):=\left(p^A(B_0,y)+p^B(B_0,y)\right)^{1/2},\qquad y\in\R.\]

\begin{thm}\label{price}
Under Assumptions \ref{initial}, \ref{M}, \ref{Markov}, \ref{scaling}, \ref{scaling1} and \ref{papb} we have $Z^{(n),B}\RA Z^B$ in $D([0,T];\R)$, where $Z^B_0:=0$ and $Z^B$ evolves as
\begin{equation*}
dZ^B(t)=\mu(Y_t)dt+\sigma(Y_t)dW^B_t,\qquad t\in[0,T],
\end{equation*}
with a standard Brownian motion $W^B$.
\end{thm}

\begin{proof}
First note that we have $Z^{(n),B}(t)=Z^{(n),B}_{\lfloor t/\Delta t^{(n)}\rfloor}$ for any $t\in[0,T]$, since $B$ is constant.
We write 
\[Z^{(n),B}(t)=Z_0^{(n),B}+\sum_{k=1}^{\lfloor t/\Delta t^{(n)}\rfloor}\delta W^{(n),B}_k+\sum_{k=1}^{\lfloor t/\Delta t^{(n)}\rfloor}\E\left(\left.\delta Z_k^{(n),B}\right|\F_{k-1}^{(n)}\right),\quad t\in[0,T],\] 
with
\[\delta W^{(n),B}_k:=\delta Z_k^{(n),B}-\E\left(\left.\delta Z_k^{(n),B}\right|\F_{k-1}^{(n)}\right).\]
First we show the convergence of the sum of conditional expectations of the increments. By definition
\[\delta Z_k^{(n),B}=\left(\Delta t^{(n)}\right)^{-1/2}\delta B_k^{(n)}=\left(\Delta t^{(n)}\right)^{-1/2}\Delta x^{(n)} \left(\1_B\left(\phi_k^{(n)}\right)-\1_A\left(\phi_k^{(n)}\right)\right)\]
and by Assumptions \ref{Markov}, \ref{scaling}, and \ref{scaling1},
\begin{eqnarray*}
\E\left(\left.\delta Z_k^{(n),B}\right|\F_{k-1}^{(n)}\right)&=&\frac{\Delta x^{(n)}}{\left(\Delta t^{(n)}\right)^{1/2}}\Delta p^{(n)}\left(p^{(n),B}\left(B_{k-1}^{(n)},Y_{k-1}^{(n)}\right)-p^{(n),A}\left(B_{k-1}^{(n)},Y_{k-1}^{(n)}\right)\right)\\
&=&\left(\Delta t^{(n)}\right)^{3/2-\alpha}\left(\Delta t^{(n)}\right)^{\alpha-1/2}p^{(n)}\left(B_{k-1}^{(n)},Y_{k-1}^{(n)}\right).
\end{eqnarray*}
Applying Theorem \ref{LLN}, Assumption \ref{papb}, and the triangle inequality one has
\[\sup_{t\leq T}\left|p^{(n)}\left(B^{(n)}(t),Y^{(n)}(t)\right)-p\left(B(t),Y(t)\right)\right|\stackrel{\p}{\longrightarrow}0,\]
which implies that indeed
\[\sum_{k=1}^{\lfloor \cdot/\Delta t^{(n)}\rfloor}\E\left(\left.\delta Z_k^{(n),B}\right|\F_{k-1}^{(n)}\right)\stackrel{\p}{\longrightarrow}\int_0^\cdot p(B_0,Y_t)dt\quad\text{in}\quad D([0,T];\R).\]
Next we show the convergence of the sum of martingale differences. Since $\beta=2(1-\alpha)$ by Assumption \ref{scaling1} we have
\begin{eqnarray*}
\E\left(\left.\left[\delta Z_k^{(n),B}\right]^2\right|\F_{k-1}^{(n)}\right)&=&\frac{\left(\Delta x^{(n)}\right)^2}{\Delta t^{(n)}}\Delta p^{(n)}\left(p^{(n),B}\left(B_{k-1}^{(n)},Y_{k-1}^{(n)}\right)+p^{(n),A}\left(B_{k-1}^{(n)},Y_{k-1}^{(n)}\right)\right)\\
&=&\Delta t^{(n)}\left(p^{(n),B}\left(B_{k-1}^{(n)},Y_{k-1}^{(n)}\right)+p^{(n),A}\left(B_{k-1}^{(n)},Y_{k-1}^{(n)}\right)\right).
\end{eqnarray*}
Now by Assumption \ref{Markov} and Theorem \ref{LLN},
\[\sup_{t\leq T}\left|p^{(n),B}\left(B^{(n)}(t),Y^{(n)}(t)\right)+p^{(n),A}\left(B^{(n)}(t),Y^{(n)}(t)\right)-p^A(B_0,Y_t)-p^B(B_0,Y_t)\right|\stackrel{\p}{\longrightarrow}0.\]
Hence for all $t\in[0,T]$,
\begin{align*}
\sum_{k=1}^{\lfloor t/\Delta t^{(n)}\rfloor}\E\left(\left.\left(\delta W^{(n),B}_k\right)^2\right|\F_{k-1}^{(n)}\right)=&\sum_{k=1}^{\lfloor t/\Delta t^{(n)}\rfloor}\E\left(\left.\left[\delta Z_k^{(n),B}\right]^2\right|\F_{k-1}^{(n)}\right)-\left(\E\left(\left.\delta Z_k^{(n),B}\right|\F_{k-1}^{(n)}\right)\right)^2\\
&
\stackrel{\p}{\longrightarrow} \int_0^tp^A(B_0,Y_u)+p^B(B_0,Y_u)du=\int_0^t\sigma^2(Y_u)du.
\end{align*}
Moreover, 
\begin{eqnarray*}
\E\left(\left.\left(\delta Z_k^{(n),B}\right)^4\right|\F_{k-1}^{(n)}\right)&=&\frac{\left(\Delta x^{(n)}\right)^4}{\left(\Delta t^{(n)}\right)^2}\Delta p^{(n)}\left(p^{(n),B}\left(B_{k-1}^{(n)},Y_{k-1}^{(n)}\right)+p^{(n),A}\left(B_{k-1}^{(n)},Y_{k-1}^{(n)}\right)\right)\\
&=&\left(\Delta t^{(n)}\right)^{2\alpha}\left(p^{(n),B}\left(B_{k-1}^{(n)},Y_{k-1}^{(n)}\right)+p^{(n),A}\left(B_{k-1}^{(n)},Y_{k-1}^{(n)}\right)\right)=o\left(\Delta t^{(n)}\right),
\end{eqnarray*}
since $\alpha>1/2$ by Assumption \ref{scaling1}. Thus, the conditional Lindeberg condition is satisfied because
\begin{eqnarray*}
\sum_{k=1}^{T_n}\E\left(\left.\left(\delta W_k^{(n),B}\right)^4\right|\F_{k-1}^{(n)}\right)&\leq&16\sum_{k=1}^{T_n}\E\left(\left.\left(\delta Z_k^{(n),B}\right)^4\right|\F_{k-1}^{(n)}\right)=o(1).
\end{eqnarray*}
Therefore, Theorem 3.33 in \cite{JS} implies the weak convergence of $\sum_{k=1}^{\cdot/\Delta t^{(n)}}\delta W_k^{(n),B},\ n\in\N$, to a Gaussian martingale with covariance function $(s,t)\mapsto\int_0^{s\wedge t}\sigma^2(Y_u)du$. Thus,
\[\sum_{k=1}^{\cdot/\Delta t^{(n)}}\delta W_k^{(n),B}\RA \int_0^\cdot \sigma(Y_t)dW^B_t,\]
where $W^B$ is a standard Brownian motion. Finally, Assumption \ref{initial} implies that $Z_0^{(n),B}\ra 0$. 
\end{proof}

\subs{Fluctuations of the volume function}

To get convergence of the volume fluctuations we need to control the second moment of order placements respectively cancelations.

\begin{ass}\label{g}
There exist measurable functions $g^{(n)}:\R\times\R\ra L^1(\R),\ n\in\N,$ 
such that for all $k=1,\dots,T_n$,
\[\Delta x^{(n)} g^{(n)}\left(B_{k-1}^{(n)},Y_{k-1}^{(n)};\cdot\right)=\E\left(\left.\1_C\left(\phi_k^{(n)}\right)\left(\omega_k^{(n)}\right)^2\1_{I^{(n)}\left(\pi_k^{(n)}\right)}(\cdot)\right|\F_{k-1}^{(n)}\right)\quad a.s.\]
Moreover, there exist $C<\infty$ and a Lipschitz continuous function $g:\R\times\R\ra L^1(\R)$ such that
\[\sup_{b,y}\left\Vert g^{(n)}(b,y)\right\Vert_\infty\leq C\quad\forall\ n\in\N\qquad\text{and}\qquad\sup_{b,y}\int_\R \left|g^{(n)}(b,y;x)-g(b,y;x)\right|dx\ra0.\]
\end{ass}

For any $\varphi\in L^2(\R)$ we define the function
\[\sigma_\varphi(y):=\left(\left(\left\langle g(B_0,y),\varphi^2\right\rangle-\left\langle f(B_0,y),\varphi\right\rangle^2\right)\vee 0\right)^{1/2},\quad y\in\R.\]
We recall that by definition for all $k=1,\dots,T_n$,
\[\delta Z_k^{(n),u}=\left(\Delta t^{(n)}\right)^{-1/2}\left(\delta u_k^{(n)}-\int_{t_{k-1}^{(n)}}^{t_k^{(n)}} f(B_0,Y_u;\cdot)du\right)\]
and
\begin{eqnarray*}
\E\left(\left.\delta u_k^{(n)}\right|\F_{k-1}^{(n)}\right)&=&\E\left(\left.\1_A\left(\phi_k^{(n)}\right)\left(T^{(n)}_--I\right)\left(u_{k-1}^{(n)}\right)+\1_B\left(\phi_k^{(n)}\right)\left(T^{(n)}_+-I\right)\left(u_{k-1}^{(n)}\right)+\Delta v^{(n)}M_k^{(n)}\right|\F_{k-1}^{(n)}\right)\\
&=&\Delta p^{(n)}p^{(n),A}\left(B_{k-1}^{(n)},Y_{k-1}^{(n)}\right)\left(T^{(n)}_--I\right)\left(u_{k-1}^{(n)}\right)\\
&&+\Delta p^{(n)}p^{(n),B}\left(B_{k-1}^{(n)},Y_{k-1}^{(n)}\right)\left(T^{(n)}_+-I\right)\left(u_{k-1}^{(n)}\right)+\Delta t^{(n)} f^{(n)}\left(B_{k-1}^{(n)},Y_{k-1}^{(n)}\right).
\end{eqnarray*}
In order to prove the convergence of $Z^{(n),u},\ n\in\N$, we make the following decomposition:
\[Z^{(n),u}_{k}=U^{(n)}_k+\sum_{j=1}^k\delta W^{(n)}_j+\sum_{j=1}^k\delta A^{(n)}_{j},\quad k=1,\dots,T_n,\] 
where
\begin{eqnarray*}
U^{(n)}_k&:=&Z_0^{(n),u}+\sum_{j=1}^{k}\left(\Delta t^{(n)}\right)^{-1/2}\left(\E\left(\left.\delta u_j^{(n)}\right|\F_{j-1}^{(n)}\right)-\Delta t^{(n)}f\left(B_{j-1}^{(n)},Y^{(n)}_{j-1}\right)\right)\\
\delta W^{(n)}_k&:=&\delta Z_k^{(n),u}-\E\left(\left.\delta Z_k^{(n),u}\right|\F_{k-1}^{(n)}\right)=\left(\Delta t^{(n)}\right)^{-1/2}
\left(\delta u_k^{(n)}-\E\left(\left.\delta u_k^{(n)}\right|\F_{k-1}^{(n)}\right)\right),\\
\delta A^{(n)}_{k}&:=&\left(\Delta t^{(n)}\right)^{-1/2}\left(\Delta t^{(n)}f\left(B_{k-1}^{(n)},Y^{(n)}_{k-1}\right)-\int_{t^{(n)}_{k-1}}^{t^{(n)}_k}f\left(B_0,Y(u)\right)du\right).
\end{eqnarray*}
By Lemma \ref{discretization} we have
\[\sup_{t\in[0,T]}\left\Vert Z^{(n),u}_{\lfloor t/\Delta t^{(n)}\rfloor}-Z^{(n),u}(t)\right\Vert_{L^2}\stackrel{\p}{\longrightarrow}0.\]
Therefore, it suffices to prove the convergence of the discrete processess $\left(Z^{(n),u}_k\right)_{k\leq T_n},\ n\in\N$, considered as piecewise constant processes on $[0,T]$, which will be done in several steps.

\begin{lem}\label{u}
Under Assumptions \ref{initial}, \ref{M}, \ref{Markov}, \ref{scaling}, \ref{scaling1}, and \ref{papb} we have for any $K>0$, 
\[\sup_{t\in[0,T]}\left|\left\langle U^{(n)}_{\lfloor t/\Delta t^{(n)}\rfloor}-\int_0^t\mu\left(Y_s\right) \partial_x u(s) ds,\varphi\right\rangle\right|=o_\p\left(\left\Vert \varphi\right\Vert_{H^3}\right).\]
Especially, $U^{(n)}$ converges in $D\left([0,T];H^{-3}\right)$ to $U:=\int_0^\cdot\mu\left(Y_s\right)\partial_x u(s) ds$.
\end{lem}

\begin{proof}
First, we see that Assumptions \ref{initial}, \ref{scaling}, and \ref{scaling1} imply that $Z_0^{(n),u}\ra 0$. Furthermore, from Assumptions \ref{scaling}, \ref{scaling1} and \ref{papb},\small
\begin{align*}
&\left(\Delta t^{(n)}\right)^{-1/2}\Delta p^{(n)}\left\langle \varphi,p^{(n),A}\left(B_{k-1}^{(n)},Y_{k-1}^{(n)}\right)\left(T^{(n)}_--I\right)\left(u_{k-1}^{(n)}\right)+p^{(n),B}\left(B_{k-1}^{(n)},Y_{k-1}^{(n)}\right)\left(T^{(n)}_+-I\right)\left(u_{k-1}^{(n)}\right)\right\rangle\\
&=\Delta t^{(n)}p^{(n)}\left(B_{k-1}^{(n)},Y_{k-1}^{(n)}\right)\left\langle \varphi,\frac{1}{\Delta x^{(n)}}\left(T^{(n)}_+-I\right)\left(u_{k-1}^{(n)}\right)\right\rangle\\
&\qquad+\left(\Delta t^{(n)}\right)^{3/2} p^{(n),A}\left(B_{k-1}^{(n)},Y_{k-1}^{(n)}\right)\left\langle \frac{1}{\Delta x^{(n)}}\left(I-T_+^{(n)}\right)(\varphi),\frac{1}{\Delta x^{(n)}}\left(T^{(n)}_+-I\right)\left(u_{k-1}^{(n)}\right)\right\rangle,
\end{align*}\normalsize
and by Lemma \ref{diffu} uniformly in $t\in[0,T]$,
\[\left\langle \varphi,\frac{1}{\Delta x^{(n)}}\left(T^{(n)}_+-I\right)\left(u^{(n)}(t)\right)\right\rangle=\langle \partial_x u(t),\varphi\rangle+o_\p\left(\left\Vert \varphi\right\Vert_{H^2}\right)\]
as well as
\[\left\langle \frac{1}{\Delta x^{(n)}}\left(I-T_+^{(n)}\right)(\varphi),\frac{1}{\Delta x^{(n)}}\left(T^{(n)}_+-I\right)\left(u^{(n)}(t)\right)\right\rangle=\left\langle \varphi'',u(t)\right\rangle+o_\p\left(\left\Vert \varphi\right\Vert_{H^3}\right).\]
Also by Assumptions \ref{Markov}, \ref{scaling}, and \ref{scaling1},\small
\begin{eqnarray*}
\sup_{k\leq T_n}\left|\left\langle \varphi,f^{(n)}\left(B^{(n)}_{k-1},Y_{k-1}^{(n)}\right)-f\left(B_{k-1}^{(n)},Y_{k-1}^{(n)}\right)\right\rangle\right|\leq \left\Vert \varphi\right\Vert_{L^2}\sup_{b,y}\left\Vert f^{(n)}\left(b,y\right)-f\left(b,y\right)\right\Vert_{L^2}=o\left(\Delta t^{(n)}\right)^{1/2}\left\Vert \varphi\right\Vert_{L^2}.
\end{eqnarray*}\normalsize
Therefore,
\[\left\langle U^{(n)}_k,\varphi\right\rangle=\Delta t^{(n)}\sum_{j=1}^{k}p^{(n)}\left(B_{j-1}^{(n)},Y_{j-1}^{(n)}\right)\left\langle \frac{1}{\Delta x^{(n)}}\left(T^{(n)}_+-I\right)\left(u_{j-1}^{(n)}\right),\varphi\right\rangle+o_\p\left(\left\Vert \varphi\right\Vert_{H^3}\right),\]
where the error term is uniform in $k\leq T_n$. Thus, the claim follows from Lemma \ref{diffu}, Assumption \ref{papb}, and Theorem \ref{LLN}.
\end{proof}

Next we analyse the martingale part of $Z^{(n),u}$. For this we fix a basis $(e_i)_{i\in \N}$ of $L^2(\R)$ and set $\sigma_i(\cdot):=\sigma_{e_i}(\cdot),\ i\in\N$.

\begin{lem}\label{wi}
Under Assumptions \ref{initial}, \ref{M}, \ref{Markov}, \ref{scaling}, \ref{scaling1}, \ref{papb}, and \ref{g} we have in $D\left([0,T];\R\times H^{-3}\right)$ the convergence \small
\[\left(\sum_{k=1}^{\lfloor\cdot/\Delta t^{(n)}\rfloor}\delta W_k^{(n),B},\sum_{k=1}^{\lfloor\cdot/\Delta t^{(n)}\rfloor}\delta W_k^{(n)}\right)\RA\left(\int_0^\cdot\sigma_B(Y_t)dW^B_t,\int_0^\cdot \sigma_B(Y_t) \partial_x u(t) dW^B_t+\sum_ie_i(x)\int_0^\cdot \sigma_i\left(Y_t\right)dW^i_t\right),\]\normalsize
where $W^B$ is a standard Brownian motion independent of all the $W^i,\ i\in\N,$ and for all $i,j\in\N$ the covariation of the standard Brownian motions $W^i$ and $W^j$ is given by
\begin{equation}\label{covariance}
\langle W^i,W^j\rangle_t=\int_0^t\frac{\langle g(B_0,Y_u), e_ie_j\rangle-\langle f(B_0,Y_u),e_i\rangle\langle f(B_0,Y_u),e_j\rangle}{\sigma_i(Y_u)\sigma_j(Y_u)}du,\quad t\in[0,T].
\end{equation}
\end{lem}

\begin{proof}
For $\varphi\in H^3$ we have
\begin{align*}
&\left(\Delta t^{(n)}\right)^{-2}\E\left(\left.\left\langle \varphi,\delta u_k^{(n)}\right\rangle^2\right|\F_{k-1}^{(n)}\right)=\left\langle \varphi, \frac{1}{\Delta x^{(n)}}\left(T_-^{(n)}-I\right)\left(u_{k-1}^{(n)}\right)\right\rangle^2p^{(n),A}\left(B_{k-1}^{(n)},Y_{k-1}^{(n)}\right)\\
&\qquad\qquad\qquad+\left\langle \varphi, \frac{1}{\Delta x^{(n)}}\left(T_+^{(n)}-I\right)\left(u_{k-1}^{(n)}\right)\right\rangle^2p^{(n),B}\left(B_{k-1}^{(n)},Y_{k-1}^{(n)}\right)+\E\left(\left.\left\langle \varphi,M_k^{(n)}\right\rangle^2\right|\F_{k-1}^{(n)}\right).
\end{align*}
We first deal with the last summand. By definition \small
\begin{eqnarray*}
\E\left(\left.\left\langle \varphi,M_k^{(n)}\right\rangle^2\right|\F_{k-1}^{(n)}\right)&=&\E\left(\left.\1_C\left(\phi_k^{(n)}\right)\left(\omega^{(n)}_k\right)^2\left(\frac{1}{\Delta x^{(n)}}\int_{I^{(n)}\left(\pi_k^{(n)}\right)}\varphi(x)dx\right)^2\right|\F^{(n)}_{k-1}\right)\\
&=&\sum_j\E\left(\left.\1_C\left(\phi_k^{(n)}\right)\left(\omega^{(n)}_k\right)^2\1_{I^{(n)}\left(\pi_k^{(n)}\right)}\left(x_j^{(n)}\right)\right|\F^{(n)}_{k-1}\right)\left(\frac{1}{\Delta x^{(n)}}\int_{I^{(n)}\left(x_j^{(n)}\right)}\varphi(x)dx\right)^2\\
&=&\sum_j \Delta x^{(n)}g^{(n)}\left(B_{k-1}^{(n)},Y_{k-1}^{(n)};x_j^{(n)}\right)\left(\frac{1}{\Delta x^{(n)}}\int_{I^{(n)}\left(x_j^{(n)}\right)}\varphi(y)dy\right)^2\\
&=&\int_\R g^{(n)}\left(B_{k-1}^{(n)},Y_{k-1}^{(n)};x\right)\left(\frac{1}{\Delta x^{(n)}}\int_{I^{(n)}(x)}\varphi(y)dy\right)^2dx
\end{eqnarray*}\normalsize
and by Assumption \ref{g}
\[\sup_{k\leq T_n}\int_\R \left(g^{(n)}\left(B_{k-1}^{(n)},Y_{k-1}^{(n)};x\right)-g\left(B_{k-1}^{(n)},Y_{k-1}^{(n)};x\right)\right)\left(\frac{1}{\Delta x^{(n)}}\int_{I^{(n)}(x)}\varphi(y)dy\right)^2dx\ra0\quad \text{a.s.}\]
Next, the Lipschitz contininuity of $g$ implies together with Theorem \ref{LLN} that
\[\sup_{t\in[0,T]}\int_\R \left(g\left(B^{(n)}(t),Y^{(n)}(t);x\right)-g\left(B_t,Y_t;x\right)\right)\left(\frac{1}{\Delta x^{(n)}}\int_{I^{(n)}(x)}\varphi(y)dy\right)^2dx\ra0.\]
Since $\varphi\in H^3$ is Lipschitz continuous and bounded, and since
\[\frac{1}{\Delta x^{(n)}}\int_{I^{(n)}(x)}\varphi(y)dy=\varphi(z_x)\quad\text{for some }z_x\in I^{(n)}(x),\]
we have
\[\left|\varphi^2(z_x)-\varphi^2(x)\right|=|\varphi(z_x)-\varphi(x)|\cdot|\varphi(z_x)+\varphi(x)|\leq C\Delta x^{(n)}\ra0\quad\text{uniformly in }x\in\R.\]
Hence, using the uniform boundedness of $g$ we conclude that
\[\sup_{t\in[0,T]}\int_\R g\left(B_0,Y_t;x\right)\left[\left(\frac{1}{\Delta x^{(n)}}\int_{I^{(n)}(x)}\varphi(y)dy\right)^2-\varphi^2(x)\right]dx\ra0\]
and therefore altogether
\[\sup_{t\in[0,T]}\left|\sum_{k=1}^{\lfloor t/\Delta t^{(n)}\rfloor}\E\left(\left.\left\langle \varphi,M_k^{(n)}\right\rangle^2\right|\F_{k-1}^{(n)}\right)-\int_0^t\left\langle g(B_0,Y_s),\varphi^2\right\rangle ds\right|\stackrel{\p}{\longrightarrow}0.\]
Moreover from Lemma \ref{diffu}, Theorem \ref{LLN}, and Assumption \ref{Markov} we deduce that uniformly in $t\in[0,T]$,
\begin{eqnarray*}
&&\left\langle \varphi, \frac{1}{\Delta x^{(n)}}\left(T_-^{(n)}-I\right)\left(u^{(n)}(t)\right)\right\rangle^2p^{(n),A}\left(B^{(n)}(t),Y^{(n)}(t)\right)\\
&&\qquad+\left\langle \varphi, \frac{1}{\Delta x^{(n)}}\left(T_+^{(n)}-I\right)\left(u^{(n)}(t)\right)\right\rangle^2p^{(n),B}\left(B^{(n)}(t),Y^{(n)}(t)\right)
\stackrel{\p}{\longrightarrow} \sigma^2(Y_t)\langle \varphi, \partial_x u\rangle^2.
\end{eqnarray*}
Finally, we have from the proof of Lemma \ref{u} the estimate 
\[\sup_{k\leq T_n}\left|\E\left(\left.\left\langle \varphi,\delta u_k^{(n)}\right\rangle\right|\F_{k-1}^{(n)}\right)-\Delta t^{(n)}\left\langle \varphi,f^{(n)}\left(B^{(n)}_{k-1},Y^{(n)}_{k-1}\right)\right\rangle\right|=\mathcal{O}_\p\left(\Delta t^{(n)}\right)^{3/2}\]
and from Theorem \ref{LLN} and Assumption \ref{Markov} we know that
\[\sup_{t\in[0,T]}\left|\left\langle \varphi,f^{(n)}\left(B^{(n)}(t),Y^{(n)}(t)\right)\right\rangle-\left\langle \varphi,f\left(B_0,Y_t\right)\right\rangle\right|=o_\p(1).\]
Furthermore, the inequality\small
\begin{align*}
&\left\langle g^{(n)}\left(B^{(n)}(t),Y^{(n)}(t)\right),\varphi^2\right\rangle-\left\langle f^{(n)}\left(B^{(n)}(t),Y^{(n)}(t)\right),\varphi\right\rangle^2=\\
&\E\left(\left.\1_C\left(\phi_k^{(n)}\right)\left(\omega_k^{(n)}\right)^2\frac{1}{\Delta x^{(n)}}\left\langle \1_{I^{(n)}\left(\pi_k^{(n)}\right)},\varphi^2\right\rangle\right|\F^{(n)}_{k-1}\right)-\left(\E\left(\left.\1_C\left(\phi_k^{(n)}\right)\omega_k^{(n)}\frac{1}{\Delta x^{(n)}}\left\langle \1_{I^{(n)}\left(\pi_k^{(n)}\right)},\varphi\right\rangle\right|\F^{(n)}_{k-1}\right)\right)^2\\
&\geq\E\left(\left.\1_C\left(\phi_k^{(n)}\right)\left(\omega_k^{(n)}\right)^2\left(\frac{1}{\Delta x^{(n)}}\left\langle \1_{I^{(n)}\left(\pi_k^{(n)}\right)},\varphi\right\rangle\right)^2\right|\F^{(n)}_{k-1}\right)-\left(\E\left(\left.\1_C\left(\phi_k^{(n)}\right)\omega_k^{(n)}\frac{1}{\Delta x^{(n)}}\left\langle \1_{I^{(n)}\left(\pi_k^{(n)}\right)},\varphi\right\rangle\right|\F^{(n)}_{k-1}\right)\right)^2\\
&=\text{Var}\left(\left.\1_C\left(\phi_k^{(n)}\right)\omega_k^{(n)}\frac{1}{\Delta x^{(n)}}\left\langle \1_{I^{(n)}\left(\pi_k^{(n)}\right)},\varphi\right\rangle\right|\F^{(n)}_{k-1}\right)\geq0
\end{align*}\normalsize
guarantees that also in the limit for all $\varphi\in H^3$,
\[\left\langle g(B_0,Y_t),\varphi^2\right\rangle-\left\langle f(B_0,Y_t),\varphi\right\rangle^2\geq0\quad\forall\ t\in[0,T].\]
Therefore for all $t\in[0,T]$,\small
\begin{align*}
&\sum_{k=1}^{\lfloor t/\Delta t^{(n)}\rfloor}\E\left(\left.\left\langle \delta W^{(n)}_k,\varphi\right\rangle^2\right|\F_{k-1}^{(n)}\right)=\Delta t^{(n)}\sum_{k=1}^{\lfloor t/\Delta t^{(n)}\rfloor}\frac{\E\left(\left.\left\langle \varphi,\delta u_k^{(n)}\right\rangle^2\right|\F_{k-1}^{(n)}\right)-\left(\E\left(\left.\left\langle \varphi,\delta u_k^{(n)}\right\rangle\right|\F_{k-1}^{(n)}\right)\right)^2}{\left(\Delta t^{(n)}\right)^2}\\
&\stackrel{\p}{\longrightarrow}\int_0^t\left[\sigma^2(Y_s)\langle \varphi, \partial_x u(s)\rangle^2+\left\langle g(B_0,Y_s),\varphi^2\right\rangle-\left\langle f(B_0,Y_s),\varphi\right\rangle^2\right] ds=
\int_0^t\left[\sigma^2(Y_s)\langle \varphi, \partial_x u(s)\rangle^2+\sigma^2_\varphi(Y_s)\right] ds.
\end{align*}\normalsize
To determine the covariation with $W^B$ we compute\small
\begin{align*}
&\E\left(\left.\delta W_k^{(n),B}\left\langle\delta W_k^{(n)},\varphi\right\rangle\right|\F_{k-1}^{(n)}\right)=
-p^{(n)}\left(B_{k-1}^{(n)},Y_{k-1}^{(n)}\right)\left(\Delta t^{(n)}\right)^{3/2}\left\langle \varphi,f^{(n)}\left(B^{(n)}_{k-1},Y_{k-1}^{(n)}\right)\right\rangle\\
&+\left(\Delta t^{(n)}-\left(\Delta t^{(n)}\right)^{1/2}\Delta x^{(n)}\Delta p^{(n)}p^{(n)}\left(B_{k-1}^{(n)},Y_{k-1}^{(n)}\right)\right)p^{(n),B}\left(B_{k-1}^{(n)},Y_{k-1}^{(n)}\right)\left\langle \varphi,\frac{1}{\Delta x^{(n)}}\left(T_+^{(n)}-I\right)\left(u_{k-1}^{(n)}\right)\right\rangle\\
&-\left(\Delta t^{(n)}+\left(\Delta t^{(n)}\right)^{1/2}\Delta x^{(n)}\Delta p^{(n)}p^{(n)}\left(B_{k-1}^{(n)},Y_{k-1}^{(n)}\right)\right)p^{(n),A}\left(B_{k-1}^{(n)},Y_{k-1}^{(n)}\right)\left\langle \varphi,\frac{1}{\Delta x^{(n)}}\left(T_-^{(n)}-I\right)\left(u_{k-1}^{(n)}\right)\right\rangle.
\end{align*}\normalsize
Since $\Delta x^{(n)}\Delta p^{(n)}=o\left(\Delta t^{(n)}\right)$ and since all of the above terms are uniformly convergent in probability, we have for all $t\in[0,T]$,
\[\sum_{k=1}^{\lfloor t/\Delta t^{(n)}\rfloor}\E\left(\left.\delta W_k^{(n),B}\left\langle \delta W_k^{(n)},\varphi\right\rangle\right|\F_{k-1}^{(n)}\right)\stackrel{\p}{\longrightarrow}\int_0^t\sigma^2\left(Y_s\right)\left\langle \varphi, \partial_x u(s)\right\rangle 
 ds.\]
Next, we check the two-dimensional conditional Lindeberg condition. Since we have already controlled the sum of fourth conditional moments of increments of $W^{(n),B}$ in the proof of Theorem \ref{price}, it suffices to control the sum of fourth conditional moments of $\delta W^{(n),\varphi}_k:=\left\langle \delta W^{(n)}_k,\varphi\right\rangle,\ k\leq T_n,$ as well. Indeed, making use of Lemma \ref{diffu} we have the following estimate, uniformly in $k=1,\dots,T_n$:
\begin{align*}
&\E\left(\left.\left(\delta W_k^{(n),\varphi}\right)^4\right|\F_{k-1}^{(n)}\right)\leq2\left(\Delta t^{(n)}\right)^{-2}
\left[\left(\Delta v^{(n)}\right)^4\E\left(\left.\left\langle M_k^{(n)},\varphi \right\rangle^4\right|\F^{(n)}_{k-1}\right)\right.\\
&\left.+\ \p\left(\left.\phi_k^{(n)}=A\right|\F^{(n)}_{k-1}\right)\left\langle \left(T_-^{(n)}-I\right)\left(u^{(n)}_{k-1}\right),\varphi \right\rangle^4+\p\left(\left.\phi_k^{(n)}=B\right|\F^{(n)}_{k-1}\right)\left\langle \left(T_+^{(n)}-I\right)\left(u^{(n)}_{k-1}\right),\varphi \right\rangle^4\right]\\
&\leq 2\left[\left(\Delta t^{(n)}\right)^2M^4\left\Vert \varphi\right\Vert^4_\infty+
\left(\Delta x^{(n)}\right)^2p^{(n),A}\left(B_{k-1}^{(n)},Y_{k-1}^{(n)}\right)\left\langle \frac{1}{\Delta x^{(n)}}\left(T_-^{(n)}-I\right)\left(u^{(n)}_{k-1}\right),\varphi \right\rangle^4\right.\\
&\qquad\qquad\qquad\left.+\ \left(\Delta x^{(n)}\right)^2p^{(n),B}\left(B_{k-1}^{(n)},Y_{k-1}^{(n)}\right)\left\langle \frac{1}{\Delta x^{(n)}}\left(T_+^{(n)}-I\right)\left(u^{(n)}_{k-1}\right),\varphi \right\rangle^4\right]=o_\p\left(\Delta t^{(n)}\right).
\end{align*}
This proves that
\[\sum_{k=1}^{T_n}\E\left(\left.\left(\delta W_k^{(n),\varphi}\right)^4\right|\F_{k-1}^{(n)}\right)=o_\p(1)\]
and hence the functional convergence theorem for martingale difference arrays (cf.~Theorem 3.33 in \cite{JS}) implies that 
\[\left(\sum_{k=1}^{\lfloor\cdot/\Delta t^{(n)}\rfloor}\delta W_k^{(n),B},\sum_{k=1}^{\lfloor\cdot/\Delta t^{(n)}\rfloor}\delta W_k^{(n),\varphi}\right)\RA\left(\int_0^\cdot\sigma(Y_t)dW^B_t,\int_0^\cdot \sigma(Y_t)\langle \partial_x u(t),\varphi\rangle dW^B_t+\int_0^\cdot \sigma_\varphi\left(Y_t\right)dW^\varphi_t\right)\]
in $D([0,T];\R\times\R)$ for independent standard Brownian motions $W^\varphi$ and $W^B$. Since $\varphi\in H^3\supset\mathcal{E}$ was arbitrary, this proves convergence of
$\left(\sum_{k=1}^{\lfloor\cdot/\Delta t^{(n)}\rfloor}\delta W_k^{(n),B},\sum_{k=1}^{\lfloor\cdot/\Delta t^{(n)}\rfloor}\delta W_k^{(n)}\right)$ in $D([0,T];\R\times\mathcal{E}')$ by Mitoma's theorem (cf.~Theorem 6.13 in \cite{Walsh}). To identify the limit correctly, we will show that for any $\varphi\in H^3\subset L^2(\R)$ with decomposition $\varphi=\sum\langle\varphi,e_i\rangle e_i$ we have the distributional equality
\[\int \sigma_\varphi(Y_t)dW^\varphi_t=\sum_i\langle \varphi,e_i\rangle \int \sigma_i(Y_t)dW_t^i,\]
where for all $i,j\in\N$ the covariation between $W^i$ and $W^j$ is given by (\ref{covariance}). 
Clearly, both processes are Gaussian local martingales. Hence, it is sufficient to prove that they have the same quadratic variation. Indeed, we have
\begin{eqnarray*}
\left\langle \sum_i\langle \varphi,e_i\rangle \int \sigma_i(Y_t)dW_t^i\right\rangle&=&\sum_{i,j}\int\sigma_i(Y_t)\langle \varphi,e_i\rangle\sigma_j(Y_t)\langle \varphi,e_j\rangle d\langle W^i,W^j\rangle_t\\
&=&\sum_{i,j}\int\langle \varphi,e_i\rangle\langle \varphi,e_j\rangle\left(\langle g(B_0,Y_t),e_ie_j\rangle-\langle f(B_0,Y_t),e_i\rangle\langle f(B_0,Y_t),e_j\rangle\right) dt\\
&=&\int \left\langle g(B_0,Y_t),\varphi^2\right\rangle-\left\langle f(B_0,Y_t),\varphi\right\rangle^2 dt\\
&=& \int \sigma^2_\varphi(Y_t)dt=\left\langle\int \sigma_\varphi(Y_t)dW^\varphi(t)\right\rangle.
\end{eqnarray*}
This proves convergence in $D([0,T];\R\times\mathcal{E}')$ to the claimed limiting process. Finally, we have to show that this convergence even holds in $D([0,T];\R\times H^{-3})$. For this let $m>0$ and $\eps>0$ be given. Then by Doob's inequality,
\begin{eqnarray*}
\p\left(\sup_{k\leq T_n}\left|\sum_{j=1}^k\delta W^{(n),\varphi}_j\right|>m\right)&\leq& m^{-2}\cdot\E\left(\sum_{j=1}^{T_n}\delta W^{(n),\varphi}_j\right)^2\leq m^{-2}\cdot\E\left(\sum_{j=1}^{T_n}\frac{1}{\Delta t^{(n)}}\left\langle \delta u^{(n)}_j,\varphi\right\rangle^2\right)\\
&\leq& \frac{\Delta t^{(n)}}{m^2}\cdot\E\sum_{j=1}^{T_n}\left(\left\langle M^{(n)}_j,\varphi\right\rangle^2+2\left\langle\frac{1}{\Delta x^{(n)}}\left(T^{(n)}_+-I\right)\left(u^{(n)}_{j-1}\right),\varphi\right\rangle^2\right).
\end{eqnarray*}
First we note that by Assumption \ref{g} for all $k\leq T_n$,
\begin{eqnarray*}
\E\left\langle \varphi,M_k^{(n)}\right\rangle^2&=&\E\left(\int_\R g^{(n)}\left(B_{k-1}^{(n)},Y_{k-1}^{(n)};x\right)\left(\frac{1}{\Delta x^{(n)}}\int_{I^{(n)}(x)}\varphi(y)dy\right)^2dx\right)\\
&\leq& \frac{1}{\Delta x^{(n)}}\E\left(\int_\R\int_\R \1_{I^{(n)}(x)}(y)g^{(n)}\left(B_{k-1}^{(n)},Y_{k-1}^{(n)};x\right)\varphi^2(y)dydx\right)\\
&\leq&\sup_{b,y}\left\Vert g^{(n)}(b,y)\right\Vert_\infty\left\Vert \varphi\right\Vert_{L^2}^2\leq C\left\Vert \varphi\right\Vert_{L^2}^2.
\end{eqnarray*}
Second, by Lemma \ref{diffu} we have
\[\Delta t^{(n)}\sum_{j=1}^{T_n}\left\langle\frac{1}{\Delta x^{(n)}}\left(T^{(n)}_+-I\right)\left(u^{(n)}_{j-1}\right),\varphi\right\rangle=\int_0^T\left\langle \partial_x u(t),\varphi\right\rangle dt + o_\p\left(\left\Vert \varphi\right\Vert_{H^3}\right).\]
Therefore, we can find $\delta>0$ such that for all $\varphi$ satisfying $\left\Vert\varphi\right\Vert_{H^3}<\delta$ we have
\[\p\left(\sup_{k\leq T_n}\left|\sum_{j=1}^k\delta W^{(n),\varphi}_j\right|>m\right)\leq \eps.\]
By Corollary 6.16 in \cite{Walsh} this proves that the convergence takes place in $D(0,T];\R\times H^{-3})$ as claimed.
\end{proof}

\begin{lem}\label{discreteouprocess}
Under Assumptions \ref{initial}, \ref{M}, \ref{Markov}, \ref{scaling}, \ref{fxx}, \ref{scaling1}, \ref{papb}, and \ref{g} there exist $L^2(\R)$-valued stochastic processes $C^{(n)}=\left(C^{(n)}_u\right)_{u\in[0,T]}$ and $D^{(n)}=\left(D^{(n)}_k\right)_{k\leq T_n}$ such that  $C^{(n)}_u$ converges to zero in probability in $L^2(\R)$, uniformly in $u\in[0,T]$, and such that for all $l\leq T_n$,
\[Z^{(n),u}_l=D^{(n)}_l+\sum_{k=1}^lZ^{(n),Y}_{k-1}\int_{t_{k-1}^{(n)}}^{t_k^{(n)}}\left(f_y\left(B_0,Y_u\right)+C^{(n)}_u\right) du\]
and
\[D^{(n)}_{\lfloor \cdot/\Delta t^{(n)}\rfloor}\RA U+\int_0^\cdot\sigma(Y_t) \partial_x u(t) dW^B_t+\sum_ie_i\int_0^\cdot\sigma_i(Y_t)dW^i_t+\int_0^\cdot f_b(B_0,Y_t) Z^B(t)dt=:D\]
in $D([0,T];H^{-3})$.
\end{lem}

\begin{proof}
For $n\in\N$ and $u\in[0,T]$, we define
\[F_u^{(n)}(s):=f\left(B_u+s\left(B_u^{(n)}-B_u\right),Y_u+s\left(Y_u^{(n)}-Y_u\right)\right),\quad s\in[0,1].\]
Then $F^{(n)}_u$ is almost surely twice differentiable in $s$ and
\begin{eqnarray*}
\frac{\partial}{\partial s}F_u^{(n)}(s)&=&\left(\Delta t^{(n)}\right)^{1/2} f_b\left(B_u+s\left(B^{(n)}(u)-B_u\right),Y_u+s\left(Y^{(n)}(u)-Y_u\right)\right) Z^{(n),B}(u)\\
&&+\left(\Delta t^{(n)}\right)^{1/2} f_y\left(B_u+s\left(B^{(n)}(u)-B_u\right),Y_u+s\left(Y^{(n)}(u)-Y_u\right)\right) Z^{(n),Y}(u).
\end{eqnarray*}
Moreover, computing the second derivative in a similar way we see that by Assumption \ref{fxx} and Theorem \ref{LLN} there exist some $L^2(\R)$-valued random variables $C^1_n(s,u),\ C_n^2(s,u)$ converging to zero in probability as $n\ra\infty$ (uniformly in $s,u$) such that
\[
\frac{\partial^2}{\partial s^2}F_u^{(n)}(s)=\left(\Delta t^{(n)}\right)^{1/2}\left(C^1_n(s,u)Z^{(n),B}(u) + C^2_n(s,u)Z^{(n),Y}(u)\right).
\]
Moreover,
\begin{eqnarray*}
\delta A_{k}^{(n)}&=&\left(\Delta t^{(n)}\right)^{-1/2}\int_{t_{k-1}^{(n)}}^{t_k^{(n)}}f\left(B^{(n)}(u),Y^{(n)}(u)\right)-f(B_0,Y_u) du\\
&=&\left(\Delta t^{(n)}\right)^{-1/2}\int_{t_{k-1}^{(n)}}^{t_k^{(n)}}\left(F^{(n)}_u(1)-F^{(n)}_u(0)\right) du
\end{eqnarray*}
and applying the fundamental theorem of calculus to $F^{(n)}_u$ gives
\begin{eqnarray*}
F_u^{(n)}(1)-F^{(n)}_u(0)=\int_0^1\frac{\partial}{\partial s}F_u^{(n)}(s)ds=\frac{\partial}{\partial s}F_u^{(n)}(0)+\int_0^1(1-s)\frac{\partial^2}{\partial s^2}F_u^{(n)}(s)ds.
\end{eqnarray*}
Therefore, we may write
\begin{align*}
\sum_{k=1}^{l}\delta A_{k}^{(n)}
&=\sum_{k=1}^{l}\int_{t_{k-1}^{(n)}}^{t_k^{(n)}}\left( f_b(B_u,Y_u)+C^{1}_n(u)\right)Z^{(n),B}(u)+\left( f_y(B_u,Y_u) +C^{2}_n(u)\right)Z^{(n),Y}(u) du\\
&=\int_0^{t_l^{(n)}}\left( f_b(B_u,Y_u)+C^{1}_n(u)\right)Z^{(n),B}(u)du+\sum_{k=1}^lZ_{k-1}^{(n),Y}\int_{t_{k-1}^{(n)}}^{t_k^{(n)}}\left( f_y(B_u,Y_u) +C^{2}_n(u)\right) du,
\end{align*}
where the $C^i_n(u),\ i=1,2$, converge to zero in probability in $L^2(\R)$, uniformly in $u\in[0,T]$, and the second equality follows from Lemma \ref{discretization}.
Concerning the first part of the above expression, Theorem \ref{price} implies that
\begin{eqnarray*}
\int_0^{\Delta t^{(n)}\lfloor t/\Delta t^{(n)}\rfloor}\left( f_b(B_u,Y_u)+C^{1}_n(u)\right)Z^{(n),B}(u)du
\RA\int_0^{\cdot}Z^{B}(u) f_b\left(B_0,Y_u\right) du.
\end{eqnarray*}
 in $D([0,T];L^2)$ and hence also in $D\left([0,T];H^{-3}\right)$. Together with Lemmata \ref{u} and \ref{wi} this yields the claim.
\end{proof}

We are now ready to prove the first main theorem of this paper, describing the convergence of the joint fluctuations of the price and volume process of the limit order book model around its first order approximation under fast rescaling (part (a) of Theorem \ref{preview}).

\begin{thm}\label{main}
Under Assumptions \ref{initial}, \ref{M}, \ref{Markov}, \ref{scaling}, \ref{h}, \ref{fxx}, \ref{scaling1}, \ref{papb}, and \ref{g}, we have weak convergence of $Z^{(n)}=\left(Z^{(n),B},Z^{(n),u}\right)$ to $(Z^B,Z^u)$ in $D\left([0,T];\R\times H^{-3}\right)$, and $(Z^B,Z^u)$ is the unique solution to
\begin{equation}\label{ISDE}
\begin{split}
Z^B(t)&=\int_0^t\mu(Y_s)ds+\int_0^t\sigma(Y_s)dW_s^B,\\
Z^u(t)&=\int_0^t\left[\mu(Y_s) \partial_x u(s)+f_b(B_0,Y_s) Z^B(s)+f_y(B_0,Y_s) \langle Z^u(s),h\rangle\right]ds\\
&\qquad\qquad +\int_0^t\sigma(Y_s) \partial_x u(s) dW^B_s+\sum_ie_i\int_0^t\sigma_i(Y_s)dW^i_s,\qquad\qquad t\in[0,T],
\end{split}
\end{equation}
where the $W^i,\ i\in\N,$ are standard Brownian motions independent of $W^B$ and for all $i,j\in\N$ the covariation of $W^i$ and $W^j$ is given in (\ref{covariance}).
\end{thm}

\begin{proof}
From Lemma \ref{discreteouprocess} we have the following representation for $Z^{(n),Y}=\langle Z^{(n),u},h\rangle$: for all $l\leq T_n$,
\[Z^{(n),Y}_l=\left\langle D^{(n)}_l,h\right\rangle+\sum_{k=1}^lZ^{(n),Y}_{k-1}\int_{t_{k-1}^{(n)}}^{t_k^{(n)}}\left\langle h,f_y\left(B_0,Y_u\right)+C^{(n)}_u\right\rangle du.\]
Hence, $\left(Z^{(n),Y}_k\right)_{k\leq T_n}$ has the structure of a discrete Ornstein-Uhlenbeck process. Following the usual solution method, we consider the discrete time processes
\begin{eqnarray*}
F_k^{(n)}&:=&\int_{t_{k-1}^{(n)}}^{t_k^{(n)}}\left\langle h,f_y\left(B_0,Y_u\right)+C^{(n)}_u\right\rangle  du,\\
X_k^{(n)}&:=&\begin{cases}Z^{(n),Y}_k\prod_{i=1}^k\left(1+F_i^{(n)}\right)^{-1}&: F_i^{(n)}> -1\  \forall\ i=1,\dots,k\\0&:\text{else}\end{cases},\qquad k\leq T_n.
\end{eqnarray*}
For each $n\in\N$ we define $\Omega_{n}:=\{F_i^{(n)}> -1\ \forall\ i=1,\dots,T_n\}$. Then $\p\left(\Omega_n\right)\ra1$ by Assumption \ref{fxx} and Lemma \ref{discreteouprocess}. On $\Omega_n$ we have for all $k\leq T_n$,
\begin{eqnarray*}
\prod_{i=1}^k\left(1+F_i^{(n)}\right)^{-1}=\exp\left(\sum_{i=1}^k\log\left(1-\frac{F_i^{(n)}}{1+F_i^{(n)}}\right)\right)=
\exp\left(-\sum_{i=1}^k\frac{F_i^{(n)}}{1+F_i^{(n)}}+Error_i^{(n)}\right)
\end{eqnarray*}
with 
\[Error_i^{(n)}=\mathcal{O}_\p\left(\frac{F_i^{(n)}}{1+F_i^{(n)}}\right)^2=\mathcal{O}_\p\left(\Delta t^{(n)}\right)^2.\]
Hence, 
\[\prod_{i=1}^{\lfloor \cdot/\Delta t^{(n)}\rfloor}\left(\left(1+F_i^{(n)}\right)^{-1}\1_{\left\{F_i^{(n}>-1\right\}}\right)\ra\exp\left(-\int_0^\cdot\langle h,f_y(B_0,Y_s)\rangle ds\right)\quad\text{in } D([0,T];\R)\]
and by similar reasoning also
\[\prod_{i=1}^{\lfloor \cdot/\Delta t^{(n)}\rfloor}\left(\left(1+F_i^{(n)}\right)\1_{\left\{F_i^{(n}>-1\right\}}\right)\ra\exp\left(\int_0^\cdot\langle h,f_y(B_0,Y_s)\rangle ds\right)\quad\text{in } D([0,T];\R).\]
Moreover, on $\Omega_n$ we have for all $k\leq T_n$,
\[\delta X_k^{(n)}=\prod_{i=1}^k\left(1+F_i^{(n)}\right)^{-1}\delta Z^{(n),Y}_k- Z_{k-1}^{(n),Y}\prod_{i=1}^{k-1}\left(1+F_i^{(n)}\right)^{-1} \frac{F_k^{(n)}}{1+F_k^{(n)}}=\prod_{i=1}^k\left(1+F_i^{(n)}\right)^{-1}\left\langle \delta D^{(n)}_k,h\right\rangle.\]
Because $\left(\langle D^{(n)},h\rangle\right)_{n\in\N}$ is C-tight by Lemma \ref{discreteouprocess}, we have joint convergence of $\left\langle D^{(n)},h\right\rangle$ and the product $\prod\left(\left(1+F_i^{(n)}\right)^{-1}\1_{\left\{F_i^{(n)}>-1\right\}}\right)$ and hence convergence of the discrete integral, i.e.
\begin{eqnarray*}
X^{(n)}(\cdot)=\sum_{i=1}^{\lfloor\cdot/\Delta t^{(n)}\rfloor}\prod_{i=1}^k\left(\left(1+F_i^{(n)}\right)^{-1}\1_{\left\{F_i^{(n)}>-1\right\}}\right)\left\langle\delta D^{(n)}_k,h\right\rangle\qquad\qquad\qquad\qquad\qquad\\
\qquad\qquad\qquad\RA \int_0^\cdot \exp\left(-\int_0^s\langle h,f_y(B_0,Y_u)\rangle du\right)d\langle D_s,h\rangle=:X(\cdot)\quad\text{in }D([0,T];\R).
\end{eqnarray*}
Especially, this shows that $\left(X^{(n)}\right)$ is also C-tight. Since $Z_k^{(n),Y}=X^{(n)}_k\prod_{i=1}^k\left(1+F_i^{(n)}\right)$ on $\Omega_n$, this implies that
\[Z^{(n),Y}_{\lfloor\cdot/\Delta t^{(n)}\rfloor}\RA X_\cdot\cdot\exp\left(\int_0^\cdot\langle h,f_y(B_0,Y_s)\rangle ds\right)=:Z^Y_\cdot\quad\text{in }D([0,T];\R),\]
 i.e.~$Z^Y$ follows the dynamics
\begin{equation}\label{OUZY}
dZ^Y(t)=d\left\langle D_t,h\right\rangle+Z^Y(t)\langle h, f_y(B_0,Y_t)\rangle dt,\quad t\in[0,T].
\end{equation}
Now the claim follows from the representation of $Z^{(n),u}$ in Lemma \ref{discreteouprocess} and the just proven convergence of $Z^{(n),Y},\ n\in\N$. Especially, (\ref{OUZY}) implies that the limiting process $(Z^B,Z^u)$ is the unique strong solution to (\ref{ISDE}). 
\end{proof}

\begin{rem}
Our second order approximation to the discrete limit order book dynamics in this section is derived under the assumption that the first order approximation of the price process is constant. This naturally rules out any fluctuations in the relative volume density function coming from fluctuations in the location variable $x$ (and not in $t$) and thus any $\partial_x Z^u$-term in the limiting dynamics.
\end{rem}

We end this section with an example illustrating Theorem \ref{main}.

\begin{ex}
We define $h$ in the same way as in Example \ref{example} and take $\phi_k^{(n)},\omega_k^{(n)},\pi_k^{(n)}$ conditionally independent. Moreover, we let the conditional distribution of $\omega_k^{(n)}$ and $\pi_k^{(n)}$ be as in Example \ref{example} and set
\begin{eqnarray*}
\p\left(\left.\phi_k^{(n)}=A\right|\F_{k-1}^{(n)}\right)&=&\frac{\Delta p^{(n)}B_{k-1}^{(n)}\left(1-\Phi\left(Y^{(n)}_{k-1}\right)\right)}{1+B^{(n)}_{k-1}},\\
\p\left(\left.\phi_k^{(n)}=B\right|\F_{k-1}^{(n)}\right)&=&\left(1+\left(\Delta t^{(n)}\right)^{1/2-\alpha}\right)\frac{\Delta p^{(n)}B_{k-1}^{(n)}\left(1-\Phi\left(Y^{(n)}_{k-1}\right)\right)}{1+B^{(n)}_{k-1}}.
\end{eqnarray*}
Then Assumption \ref{papb} is satisfied with $p^{(n)}(b,y)=\frac{b}{1+b}$ and Assumption \ref{g} is satisfied with
\[g^{(n)}(b,y;x)=\left(1-\Delta p^{(n)}\right)\frac{C}{\Delta x^{(n)}}\sum_{j\in\Z}\1_{\left(x_j^{(n)},x_{j+1}^{(n)}\right]}(x)\int_{x_j^{(n)}}^{x^{(n)}_{j+1}}(z-10)^2(z+10)^2\1_{[-10,10]}(z)dz.\]
Therefore, in this example the first order approximation is given by 
\begin{eqnarray*}
B_t=B_0,\quad u(t,x)=u_0(x)+C(x-10)^2(x+10)^2\1_{[-10,10]}(x)\int_0^t\1_{[0,\infty)}(Y_s)6Y_s^2e^{-Y_s^3}ds
\end{eqnarray*}
and the second order approximation satisfies for any $\varphi\in H^3$,\small
\begin{align*}
&\qquad dZ^B(t)=\frac{B_0(1-\Phi(Y_t))}{1+B_0}dt+\sqrt{\frac{2B_0(1-\Phi(Y_t))}{1+B_0}}dW^B_t\\
& d\langle Z^u(t),\varphi\rangle=\1_{[0,\infty)}(Y_t)6Y_t^2e^{-Y_t^3}\langle Z^u(t),h\rangle C\int_{-10}^{10}(x-10)^2(x+10)^2\varphi(x)dx dt+\left\langle \partial_x u(t),\varphi\right\rangle dZ^B(t)\\
&\quad+\1_{[0,\infty)}(Y_t)6Y_t^2e^{-Y_t^3}C\sqrt{\int_{-10}^{10}(x-10)^2(x+10)^2\varphi^2(x)dx-\left(\int_{-10}^{10}(x-10)^2(x+10)^2\varphi(x)dx\right)^2}dW_t^\varphi.
\end{align*}\normalsize
Especially, high standing volumes at the top of the book will lead to a smaller drift and smaller volatility of price fluctuations.
\end{ex}

\s{Renormalization on a slow time scale}\label{slow}

In this section we will rescale by a slower time scale corresponding to the frequency of price changes.

\begin{ass}\label{scaling2}
\[\Delta^{(n)}=\Delta x^{(n)},\qquad\alpha\in\left(0,\frac{1}{2}\right),\qquad \beta=1-\alpha.\]
\end{ass}

Since $\beta=1-\alpha$ under Assumption \ref{scaling2}, the limiting price process $B$ is not constant in this case and depends on the path of $Y$ as well. Hence, we cannot analyze the fluctuations $Z^{(n),B}$ independently of $Z^{(n),Y}$ as we did in the previous section. For the analysis we will again decompose $Z^{(n),B}$ as a semimartingale, i.e. we write for $k=1,\dots,T_n$,
\begin{eqnarray*}
Z^{(n),B}_k&=&Z_0^{(n),B}+N_k^{(n)}+\sum_{j=1}^k\E\left(\left.\delta Z_j^{(n),B}\right|\F_{j-1}^{(n)}\right)
\end{eqnarray*}
with
\[N_k^{(n)}:=\sum_{j=1}^k\delta W^{(n)}_j\qquad\text{and}\qquad\delta W^{(n)}_j:=\delta Z_j^{(n),B}-\E\left(\left.\delta Z_j^{(n),B}\right|\F_{j-1}^{(n)}\right).\]
On the other hand, writing the dynamics of $Z^{(n),u}$ in semimartingale form has the disadvantage that it involves the ``discrete derivative'' of $Z^{(n),u}$ itself with respect to $x$. To analyze it we would therefore need some prior knowledge on the speed of convergence of $u^{(n)}$ to $u$ - however, that is precisely what we would like to find out by establishing a central limit theorem! Therefore, we will choose another non-semimartingale decomposition for $Z^{(n),u}$. For this we first define the absolute volume density functions
\begin{equation*}
v^{(n)}_k(x):=u^{(n)}_k\left(x-B^{(n)}_k\right)=v_0^{(n)}(x)+\sum_{j=1}^k \Delta v^{(n)}M^{(n)}_j\left(x-B_{j-1}^{(n)}\right),\quad k\leq T_n,\ n\in\N,
\end{equation*}
and
\begin{equation}\label{uvu}
v(t,x):=u(t,x-B_t)=v(0,x)+\int_0^tf(B_u,Y_u;x-B_u)du,\quad (t,x)\in[0,T]\times\R.
\end{equation}
This allows us to write\small
\begin{equation*}\label{Znu}
\begin{split}
Z^{(n),u}_k(x)&=\frac{u^{(n)}_k(x)-u\left(t_k^{(n)},x\right)}{\left(\Delta x^{(n)}\right)^{1/2}}=\frac{v^{(n)}_k\left(x+B_k^{(n)}\right)-v\left(t_k^{(n)},x+B_{t_k^{(n)}}\right)}{\left(\Delta x^{(n)}\right)^{1/2}}\\
&=X^{(n),0}_k(x)+\left(\Delta x^{(n)}\right)^{-1/2}\left(\sum_{j=1}^k \Delta v^{(n)}M^{(n)}_j\left(x+B^{(n)}_k-B_{j-1}^{(n)}\right)-\int_0^{t_k^{(n)}}f\left(B_u,Y_u;x+B_{t_k^{(n)}}-B_u\right)du\right)\\
&=X^{(n),0}_k(x)+X^{(n),1}_k(x)+X^{(n),2}_k(x)
\end{split}
\end{equation*}\normalsize
with\small
\begin{eqnarray*}
X^{(n),0}_k(\cdot)&:=&\frac{v_0^{(n)}\left(\cdot+B_k^{(n)}\right)-v\left(0,\cdot+B_{t_k^{(n)}}\right)}{\left(\Delta x^{(n)}\right)^{1/2}}=\frac{u_0^{(n)}\left(\cdot-B_0^{(n)}+B_k^{(n)}\right)-u_0\left(\cdot-B_0+B_{t^{(n)}_k}\right)}{\left(\Delta x^{(n)}\right)^{1/2}},\\
X^{(n),1}_k(\cdot)&:=&\left(\Delta x^{(n)}\right)^{-1/2}\sum_{j=1}^k\Delta v^{(n)}M_j^{(n)}\left(\cdot-B^{(n)}_{j-1}+B_{t_k^{(n)}}\right)-\Delta t^{(n)}f^{(n)}\left(B_{j-1}^{(n)},Y_{j-1}^{(n)};\cdot-B_{j-1}^{(n)}+B_{t^{(n)}_k}\right),\\
X^{(n),2}_k(\cdot)&:=&\left(\Delta x^{(n)}\right)^{-1/2}\left(\Delta t^{(n)}\sum_{j=1}^k f^{(n)}\left(B^{(n)}_{j-1},Y^{(n)}_{j-1};\cdot+B^{(n)}_k-B_{j-1}^{(n)}\right)
-\int_0^{t_k^{(n)}}f\left(B_u,Y_u;\cdot-B_u+B_{t_k^{(n)}}\right)du\right).
\end{eqnarray*}\normalsize

In order to prove convergence of the pair $\left(Z^{(n),B},Z^{(n),u}\right)$ we will analyze each part in the above decompositions separately. The next two lemmata deal with the martingale part $N^{(n)}$ of the price fluctuations and the ``shifted martingale" part $X^{(n),1}$ of the volume process. 

\begin{lem}\label{N}
Under Assumptions \ref{initial}, \ref{M}, \ref{Markov}, \ref{scaling}, and \ref{scaling2},
\[N^{(n)}\RA \int_0^\cdot \sigma_B(B_t,Y_t)dW_t\quad\text{in}\quad D\left([0,T];\R\right),\]
where $W$ is a standard Brownian motion.
\end{lem}

\begin{proof}
We have for all $l\leq T_n$,
\begin{eqnarray*}
\E\left(\left.\left[\delta W_k^{(n)}\right]^2\right|\F_{k-1}^{(n)}\right)&=&\Delta x^{(n)}\Delta p^{(n)}\left(p^{(n),B}\left(B_{k-1}^{(n)},Y_{k-1}^{(n)}\right)+p^{(n),A}\left(B_{k-1}^{(n)},Y_{k-1}^{(n)}\right)\right)\\
&&-\frac{\left(\Delta t^{(n)}\right)^2}{\Delta x^{(n)}}\left(p^{(n),B}\left(B_{k-1}^{(n)},Y_{k-1}^{(n)}\right)-p^{(n),A}\left(B_{k-1}^{(n)},Y_{k-1}^{(n)}\right)\right)^2.
\end{eqnarray*}
Now by Assumption \ref{Markov} and Theorem \ref{LLN},
\[\sup_{t\leq T}\left|p^{(n),B}\left(B^{(n)}(t),Y^{(n)}(t)\right)+p^{(n),A}\left(B^{(n)}(t),Y^{(n)}(t)\right)-p^{A}\left(B_t,Y_t\right)-p^{B}\left(B_t,Y_t\right)\right|\stackrel{\p}{\longrightarrow}0.\]
Hence for all $t\in[0,T]$,
\begin{align*}
&\sum_{k=1}^{\lfloor t/\Delta t^{(n)}\rfloor}\E\left(\left.\left(\delta W^{(n)}_k\right)^2\right|\F_{k-1}^{(n)}\right)
\stackrel{\p}{\longrightarrow} \int_0^tp^{A}(B_u,Y_u)+p^{B}(B_u,Y_u)du=\int_0^t\sigma_B^2(B_u,Y_u)du.
\end{align*}
Moreover, the conditional Lindeberg condition is satisfied, since uniformly in $k\leq T_n$,
\begin{eqnarray*}
\E\left(\left.\left(\delta W_k^{(n)}\right)^4\right|\F_{k-1}^{(n)}\right)&\leq&16\left(\Delta x^{(n)}\right)^2\Delta p^{(n)}\left(p^{(n),B}\left(B_{k-1}^{(n)},Y_{k-1}^{(n)}\right)+p^{(n),A}\left(B_{k-1}^{(n)},Y_{k-1}^{(n)}\right)\right)=o\left(\Delta t^{(n)}\right).
\end{eqnarray*}
Therefore, Theorem 3.33 in \cite{JS} implies the weak convergence of $N^{(n)},\ n\in\N$, to a Gaussian martingale with covariance function $(s,t)\mapsto\int_0^{s\wedge t}\sigma_B^2(B_u,Y_u)du$. Thus, denoting by $W$ a Brownian motion,
\[N^{(n)}\RA \int_0^\cdot \sigma_B(B_t,Y_t)dW_t.\]
\end{proof}

\begin{lem}\label{x1}
Under Assumptions \ref{M}, \ref{Markov}, \ref{scaling}, and \ref{scaling2},  we have $\sup_{k\leq T_n}\left\Vert X^{(n),1}_k\right\Vert_{L^2}\stackrel{\p}{\longrightarrow}0$.
\end{lem}

\begin{proof}
Let $\eps>0$. First, we note that even though $\left(X^{(n),1}_k\right)_{k\leq T_n}$ is not a martingale, the shifted process $\left(X^{(n),1}_k\left(\cdot-B_{t_k^{(n)}}\right)\right)_{k\leq T_n}$ is actually a martingale. Hence, we can apply 
Lemma \ref{Pisier} to the shifted process. Thus, we have\small
\begin{eqnarray*}
\p\left(\sup_{k\leq T_n}\left\Vert X^{(n),1}_k\right\Vert_{L^2}>\eps\right)&=&\p\left(\sup_{k\leq T_n}\left\Vert X^{(n),1}_k\left(\cdot-B_{t^{(n)}_k}\right)\right\Vert_{L^2}>\eps\right)\leq \eps^{-1}\E\left(\sup_{k\leq T_n}\left\Vert X_k^{(n),1}\left(\cdot-B_{t^{(n)}_k}\right)\right\Vert_{L^2}\right)\\
&\leq& C_\eps\cdot\E\left(\sum_{k=1}^{T_n}\left\Vert X^{(n),1}_k\left(\cdot-B_{t^{(n)}_k}\right)-X^{(n),1}_{k-1}\left(\cdot-B_{t^{(n)}_{k-1}}\right)\right\Vert_{L^2}^2\right)^{1/2}\\
&\leq& \frac{C_\eps\Delta t^{(n)}}{\left(\Delta x^{(n)}\right)^{1/2}}\left(\sum_{k=1}^{T_n}\E\left\Vert M^{(n)}_k\left(\cdot-B_{k-1}^{(n)}\right)-f^{(n)}\left(B^{(n)}_{k-1},Y^{(n)}_{k-1};\cdot-B^{(n)}_{k-1}\right)\right\Vert_{L^2}^2\right)^{1/2}\\
&\leq& \frac{C_\eps\Delta t^{(n)}T_n^{1/2}}{\left(\Delta x^{(n)}\right)^{1/2}}\left(\sup_{k\leq T_n}\E\left\Vert M^{(n)}_k\right\Vert_{L^2}^2\right)^{1/2}\leq C_\eps M\frac{\sqrt{T\Delta t^{(n)}}}{\Delta x^{(n)}}=\mathcal{O}\left(\Delta t^{(n)}\right)^{1/2-\alpha}.
\end{eqnarray*}\normalsize
Since $\eps>0$ was arbitrary and since $\alpha<\frac{1}{2}$ by Assumption \ref{scaling2}, we conclude that \mbox{$\sup_{k\leq T_n}\left\Vert X^{(n),1}_k\right\Vert_{L^2}\stackrel{\p}{\longrightarrow}0$.}
\end{proof}

The next three lemmata give approximations for the drift part of $Z^{(n),B}$ and for the two processes $X^{(n),0}$ and $X^{(n),2}$. They allow us to guess the limiting dynamics of $Z^{(n),B}$ and $\left\langle Z^{(n),u},\varphi\right\rangle$ for  $\varphi\in H^3$.

\begin{lem}\label{ZnB}
Under Assumptions \ref{initial}, \ref{M}, \ref{Markov}, \ref{scaling}, \ref{pxx}, and \ref{scaling2}, there exist stochastic processes $C^{(n),1}=\left(C^{(n),1}_t\right),\ C^{(n),2}:=\left(C^{(n),2}_t\right)$, converging to zero in probability as $n\ra\infty$, uniformly in $t\in[0,T]$, such that for all $n\in\N$ and $k\leq T_n$,
\begin{equation*}
Z^{(n),B}_k=N^{(n)}_k+\int_0^{t^{(n)}_k}\left(p_b(B_t,Y_t)+C^{(n),1}_t\right)Z^{(n),B}_t+\left(p_y(B_t,Y_t)+C^{(n),2}_t\right)Z^{(n),Y}_tdt+C^{(n),3}_k,
\end{equation*}
where $\sup_{k\leq T_n}C^{(n),3}_k\ra0$ a.s. 
\end{lem}

\begin{proof}
First we note that Assumptions \ref{initial}, \ref{scaling}, and \ref{scaling2} imply that $Z_0^{(n),B}\ra 0$. Moreover, let us write $p^{B-A}:=p^B-p^A$ and $p^{(n),B-A}:=p^{(n),B}-p^{(n),A}$. Then by definition\small
\[\delta Z_k^{(n),B}=\left(\Delta x^{(n)}\right)^{-1/2}\delta B_k^{(n)}=\left(\Delta x^{(n)}\right)^{1/2} \left(\1_B\left(\phi_k^{(n)}\right)-\1_A\left(\phi_k^{(n)}\right)\right)-\left(\Delta x^{(n)}\right)^{-1/2}\int_{t_{k-1}^{(n)}}^{t_k^{(n)}}p^{B-A}(B_t,Y_t)dt\]\normalsize
and by Assumptions \ref{Markov}, \ref{scaling}, and \ref{scaling2} we can write
\begin{eqnarray*}
\E\left(\left.\delta Z_k^{(n),B}\right|\F_{k-1}^{(n)}\right)&=&\left(\Delta x^{(n)}\right)^{-1/2}\int_{t_{k-1}^{(n)}}^{t_k^{(n)}}\left(p^{(n),B-A}\left(B^{(n)}_t,Y^{(n)}_t\right)-p^{B-A}(B_t,Y_t)\right)dt\\
&=&C^{(n),3}_k+\left(\Delta x^{(n)}\right)^{-1/2}\int_{t_{k-1}^{(n)}}^{t_k^{(n)}}\left(p^{B-A}\left(B^{(n)}_t,Y^{(n)}_t\right)-p^{B-A}\left(B_t,Y_t\right)\right)dt
\end{eqnarray*}
with $\sup_{k\leq T_n}C^{(n),3}_k\ra0$ a.s. Furthermore, due to Assumption \ref{pxx} and Theorem \ref{LLN} there exist random variables $C^{(n),i}_{u,t},\ i=1,2$, converging to zero in probability (uniformly in $u,t$) such that
\begin{align*}
&\left(\Delta x^{(n)}\right)^{-1/2}\left(p^{B-A}\left(B^{(n)}_t,Y^{(n)}_t\right)-p^{B-A}(B_t,Y_t)\right)\\
&\qquad\qquad=\int_0^1p_b^{B-A}\left(B_t+s\left(B_t^{(n)}-B_t\right),Y_t+s\left(Y_t^{(n)}-Y_t\right)\right)Z^{(n),B}_tds\\
&\qquad\qquad\qquad+\int_0^1p_y^{B-A}\left(B_t+s\left(B_t^{(n)}-B_t\right),Y_t+s\left(Y_t^{(n)}-Y_t\right)\right)Z^{(n),Y}_tds\\
&\qquad\qquad=p_b^{B-A}(B_t,Y_t)Z_t^{(n),B}+p_y^{B-A}(B_t,Y_t)Z_t^{(n),Y}+\int_0^1\int_0^s C^{(n),1}_{u,t}Z_t^{(n),B}+C^{(n),2}_{u,t}Z_t^{(n),Y}duds.
\end{align*}
Therefore, we may write 
\begin{equation*}
\sum_{j=1}^k\E\left(\left.\delta Z^{(n),B}_k\right|\F_{k-1}^{(n)}\right)= \int_0^{t^{(n)}_k}\left(p_b(B_t,Y_t)+C^{(n),1}_t\right)Z^{(n),B}_t+\left(p_y(B_t,Y_t)+C^{(n),2}_t\right)Z^{(n),Y}_tdt+C_k^{(n),3},
\end{equation*}
where $C^{(n),1}_t, C^{(n),2}_t$ are converging to zero in probability as $n\ra\infty$ uniformly in \mbox{$t\in[0,T]$.}
\end{proof}

\begin{lem}\label{x0}
Under Assumptions \ref{initial}, \ref{M}, \ref{Markov}, \ref{scaling}, and \ref{scaling2}, there exist for any $\varphi\in H^3$ random variables $C^{(n)}_k,\ k\leq T_n$, converging to zero in probability, uniformly in $k\leq T_n$, such that
\[\left\langle X_k^{(n),0},\varphi\right\rangle=
\left(\left\langle u_0'\left(\cdot-B_0+B_{t_k^{(n)}}\right),\varphi\right\rangle+C^{(n)}_k\right)
\left(Z^{(n),B}_k-Z^{(n),B}_0\right)\quad\text{with}\quad\sup_{k\leq T_n}\left|C^{(n)}_k\right|=o_\p\left(\left\Vert \varphi\right\Vert_{H^2}\right).\]
\end{lem}

\begin{proof}
We have
\begin{eqnarray*}
\left\langle \varphi,X_k^{(n),0}\right\rangle&=&
\left(\Delta x^{(n)}\right)^{-1/2}\left\langle\varphi, u_0^{(n)}\left(\cdot-B_0^{(n)}+B_k^{(n)}\right)-u_0\left(\cdot-B_0^{(n)}+B_k^{(n)}\right)\right\rangle\\
&&\quad+\left(\Delta x^{(n)}\right)^{-1/2}\left\langle\varphi, u_0\left(\cdot-B_0^{(n)}+B_k^{(n)}\right)-u_0\left(\cdot-B_0+B_{t_k^{(n)}}\right)\right\rangle
\end{eqnarray*}
and 
\[\left(\Delta x^{(n)}\right)^{-1/2}\left|\left\langle\varphi, u_0^{(n)}\left(\cdot-B_0^{(n)}+B_k^{(n)}\right)-u_0\left(\cdot-B_0^{(n)}+B_k^{(n)}\right)\right\rangle\right|\leq \left\Vert\varphi\right\Vert_{L^2}\left(\Delta x^{(n)}\right)^{-1/2}\left\Vert u_0^{(n)}-u_0\right\Vert_{L^2},\]
which converges to zero by Assumption \ref{initial}. Moreover, 
\begin{align*}
&\left(\Delta x^{(n)}\right)^{-1/2}\left\langle\varphi, u_0\left(\cdot-B_0^{(n)}+B_k^{(n)}\right)-u_0\left(\cdot-B_0+B_{t_k^{(n)}}\right)\right\rangle\\
=&\left\langle\varphi,\int_0^1u_0'\left(\cdot-B_0+B_{t_k^{(n)}}-s\left(B_0^{(n)}-B_0-B^{(n)}_k+B_{t_k^{(n)}}\right)\right)ds\right\rangle
\left(Z^{(n),B}_k-Z^{(n),B}_0\right)\\
=&\left(\left\langle\varphi,u_0'\left(\cdot-B_0+B_{t_k^{(n)}}\right)\right\rangle+C^{(n)}_k\right)
\left(Z^{(n),B}_k-Z^{(n),B}_0\right)
\end{align*}
with\small
\[C^{(n)}_k:=\left(B_0-B_0^{(n)}+B_k^{(n)}-B_{t_k^{(n)}}\right)\int_0^1\int_0^s\left\langle \varphi'',u_0\left(\cdot-B_0+B_{t_k^{(n)}}-r\left(B_0^{(n)}-B_0-B^{(n)}_k+B_{t_k^{(n)}}\right)\right)\right\rangle drds\]\normalsize
converging to zero in probability, uniformly in $k\leq T_n$, by Theorem \ref{LLN}. 
\end{proof}

\begin{lem}\label{x2}
Under Assumptions \ref{initial}, \ref{M}, \ref{Markov}, \ref{scaling}, \ref{fxx}, and \ref{scaling2}, there exist for any $\varphi\in H^3$ random variables $C^{(n),3}_k,C^{(n),i}_k(u),\ i=1,2,\ u\leq t_k^{(n)},\ k\leq T_n,$ satisfying
\[\sup_{k\leq T_n}\left|C^{(n),3}_k\right|=o_\p\left(\left\Vert\varphi\right\Vert_{H^2}\right)\qquad\text{and}\qquad
\sup_{k\leq T_n}\sup_{u\leq t_k^{(n)}}\left|C^{(n),i}_k(u)\right|=o_\p\left(\left\Vert\varphi\right\Vert_{H^2}\right),\quad i=1,2,\] 
such that
\begin{align*}
&\left\langle X^{(n),2}_k,\varphi\right\rangle=
\int_0^{t_k^{(n)}}\left(\left\langle f_y\left(B_u,Y_u\right),\varphi\left(\cdot+B_u-B_{t_k^{(n)}}\right)\right\rangle+C^{(n),2}_k(u)\right) Z^{(n),Y}(u)du\\ 
&\ +\int_0^{t_k^{(n)}} \left(\left\langle f_b\left(B_u,Y_u\right),\varphi\left(\cdot+B_u-B_{t_k^{(n)}}\right)\right\rangle+\left\langle f\left(B_u,Y_u\right),\varphi'\left(\cdot+B_u-B_{t_k^{(n)}}\right)\right\rangle+C^{(n),1}_k(u)\right)Z^{(n),B}(u)du\\
&\ -Z^{(n),B}_k\left(\left\langle u_0,\varphi'\left(\cdot+B_0-B_{t_k^{(n)}}\right)\right\rangle+\int_0^{t_k^{(n)}}\left\langle f\left(B_u,Y_u\right),\varphi'\left(\cdot+B_u-B_{t_k^{(n)}}\right)\right\rangle du +C^{(n),3}_k\right)+o\left(\left\Vert\varphi\right\Vert_{L^2}\right),
\end{align*}
where the error term is uniform in $k\leq T_n$.
\end{lem}

\begin{proof}
First, we note that we may replace $f^{(n)}$ with $f$ in the definition of $X^{(n),2}$. Indeed,
\begin{align*}
&\sup_{k\leq T_n}\frac{\Delta t^{(n)}}{\left(\Delta x^{(n)}\right)^{1/2} }\left|\sum_{j=1}^k\left\langle \varphi, f^{(n)}\left(B^{(n)}_{j-1},Y^{(n)}_{j-1};\cdot+B^{(n)}_k-B_{j-1}^{(n)}\right)-f\left(B^{(n)}_{j-1},Y^{(n)}_{j-1};\cdot+B^{(n)}_k-B_{j-1}^{(n)}\right)\right\rangle\right|\\
&\leq\frac{\Delta t^{(n)}\left\Vert\varphi\right\Vert_{L^2}}{\left(\Delta x^{(n)}\right)^{1/2} }\sum_{j=1}^{T_n}\left\Vert f^{(n)}\left(B^{(n)}_{j-1},Y^{(n)}_{j-1}\right)-f\left(B^{(n)}_{j-1},Y^{(n)}_{j-1}\right)\right\Vert_{L^2}
\leq 
\frac{T\left\Vert\varphi\right\Vert_{L^2}}{\left(\Delta x^{(n)}\right)^{1/2} }\sup_{b,y}\left\Vert f^{(n)}\left(b,y\right)-f\left(b,y\right)\right\Vert_{L^2},
\end{align*}
which converges to zero by Assumption \ref{Markov}. Next, we define for $\varphi\in H^3$, $n\in\N,\ k\leq T_n$, $u\leq t_k^{(n)}$, and $s\in[0,1]$,
\begin{align*}
F_{u,k}^{(n)}(s):=\left\langle f\left(B_u+s\left(B^{(n)}(u)-B_u\right),Y_u+s\left(Y^{(n)}(u)-Y_u\right);\cdot+B_{t_k^{(n)}}+s\left(B_k^{(n)}-B_{t_k^{(n)}}\right)\right),\quad\qquad\right.\\
\qquad\qquad\left.\varphi\left(\cdot+B_u+s\left(B^{(n)}(u)-B_u\right)\right)\right\rangle.
\end{align*}
Then,
\begin{eqnarray*}
\int_0^t\left\langle f\left(B^{(n)}(u),Y^{(n)}(u);\cdot+B^{(n)}_k-B^{(n)}(u)\right)-f\left(B_u,Y_u;\cdot+B_{t^{(n)}_k}-B_u\right),\varphi\right\rangle du=\int_0^tF^{(n)}_{u,k}(1)-F^{(n)}_{u,k}(0)du
\end{eqnarray*}
and by the fundamental theorem of calculus 
\begin{eqnarray*}
F_{u,k}^{(n)}(1)-F^{(n)}_{u,k}(0)=\int_0^1\frac{\partial}{\partial s}F_{u,k}^{(n)}(s)ds=\frac{\partial}{\partial s}F_{u,k}^{(n)}(0)+\int_0^1(1-s)\frac{\partial^2}{\partial s^2}F_{u,k}^{(n)}(s)ds.
\end{eqnarray*}
Clearly, $F^{(n)}_{u,k}(s)$ is almost surely twice differentiable in $s$ with first derivative $\frac{\partial}{\partial s}F_{u,k}^{(n)}(s)$ equal to \scriptsize
\begin{align*}
&\left\langle f_b\left(B_u+s\left(B^{(n)}(u)-B_u\right),Y_u+s\left(Y^{(n)}(u)-Y_u\right)\right),\varphi\left(\cdot+B_u+s\left(B^{(n)}(u)-B_u\right)-B_{t^{(n)}_k}-s\left(B^{(n)}_k-B_{t^{(n)}_k}\right)\right)\right\rangle\left(B^{(n)}(u)-B_u\right)\\
&+\left\langle f\left(B_u+s\left(B^{(n)}(u)-B_u\right),Y_u+s\left(Y^{(n)}(u)-Y_u\right)\right),\varphi'\left(\cdot+B_u+s\left(B^{(n)}(u)-B_u\right)-B_{t^{(n)}_k}-s\left(B^{(n)}_k-B_{t^{(n)}_k}\right)\right)\right\rangle\left(B^{(n)}(u)-B_u\right)\\
&-\left\langle f\left(B_u+s\left(B^{(n)}(u)-B_u\right),Y_u+s\left(Y^{(n)}(u)-Y_u\right)\right),\varphi'\left(\cdot+B_u+s\left(B^{(n)}(u)-B_u\right)-B_{t_k^{(n)}}-s\left(B^{(n)}_k-B_{t_k^{(n)}}\right)\right)\right\rangle\left(B^{(n)}_k-B_{t^{(n)}_k}\right)\\
&+\left\langle f_y\left(B_u+s\left(B^{(n)}(u)-B_u\right),Y_u+s\left(Y^{(n)}(u)-Y_u\right)\right),\varphi\left(\cdot+B_u+s\left(B^{(n)}(u)-B_u\right)-B_{t_k^{(n)}}-s\left(B^{(n)}_k-B_{t^{(n)}_k}\right)\right)\right\rangle\left(Y^{(n)}(u)-Y_u\right).
\end{align*}\normalsize
Moreover, by Assumption \ref{fxx}, Theorem \ref{LLN}, there exist some random variables $C^{(n),i}_k(s,u),\ i=1,2,3,$ satisfying $\sup_{s\in[0,1]}\sup_{k\leq T_n}\sup_{u\leq t_k^{(n)}}\left|C^{(n),i}_k(s,u)\right|=o_\p\left(\left\Vert\varphi\right\Vert_{H^2}\right)$ such that
\[
\frac{\partial^2}{\partial s^2}F_{u,k}^{(n)}(s)=\sqrt{\Delta x^{(n)}}\left(C^{(n),1}_k(s,u)Z^{(n),B}(u) + C^{(n),2}_k(s,u)Z^{(n),Y}(u)+ C^{(n),3}_k(s,u)Z^{(n),B}_k\right).
\]
Therefore, we may write
\begin{eqnarray*}
&&\frac{1}{\sqrt{\Delta x^{(n)}}}\int_0^tF^{(n)}_{u,k}(1)-F^{(n)}_{u,k}(0)du=\\
&&\qquad\qquad\int_0^t\left(\frac{1}{\sqrt{\Delta x^{(n)}}}\frac{\partial}{\partial s}F_{u,k}^{(n)}(0)+C^{(n),1}_k(u)Z^{(n),B}(u)+C^{(n),2}_k(u)Z^{(n),Y}(u) \right)du+C^{(n),3}_kZ^{(n),B}_k,
\end{eqnarray*}
where $C^{(n),1}_k(u),\ C^{(n),2}_k(u),\ C^{(n),3}_k$ converge to zero in probability (uniformly in $u,k$) and
\begin{eqnarray*}
&&\frac{1}{\sqrt{\Delta x^{(n)}}}\frac{\partial}{\partial s}F_{u,k}^{(n)}(0)=
\left\langle f\left(B_u,Y_u\right),\varphi'\left(\cdot+B_u-B_{t^{(n)}_k}\right)\right\rangle\left(Z^{(n),B}(u)-Z^{(n),B}_k\right)\\
&&\qquad+\left\langle f_b\left(B_u,Y_u\right),\varphi\left(\cdot+B_u-B_{t^{(n)}_k}\right)\right\rangle Z^{(n),B}(u)+\left\langle f_y\left(B_u,Y_u\right),\varphi\left(\cdot+B_u-B_{t^{(n)}_k}\right)\right\rangle Z^{(n),Y}(u).
\end{eqnarray*}
\end{proof}

Note that Lemma \ref{x1} implies that $\sup_{k\leq T_n}\left\langle \varphi,X^{(n),1}_k\right\rangle\stackrel{\p}{\longrightarrow}0$ for any $\varphi\in H^3\subset L^2$. Therefore, we have from Lemmata \ref{x0}, \ref{x1}, and \ref{x2} for any $\varphi\in H^3$ the following representation: for all $n\in\N$ and $k\leq T_n$,\small
\begin{align*}
\left\langle Z^{(n),u}_k,\varphi\right\rangle&=C^{(n)}_k+\int_0^{t_k^{(n)}}\left(\left\langle f_y\left(B_u,Y_u\right),\varphi\left(\cdot+B_u-B_{t_k^{(n)}}\right)\right\rangle+C^{(n),2}_k(u)\right) Z^{(n),Y}_udu\\ 
&+\int_0^{t_k^{(n)}} \left(\left\langle f_b\left(B_u,Y_u\right),\varphi\left(\cdot+B_u-B_{t_k^{(n)}}\right)\right\rangle+\left\langle f\left(B_u,Y_u\right),\varphi'\left(\cdot+B_u-B_{t_k^{(n)}}\right)\right\rangle+C^{(n),1}_k(u)\right)Z^{(n),B}_udu\\
&-Z^{(n),B}_k\left(\left\langle u_0,\varphi'\left(\cdot+B_0-B_{t_k^{(n)}}\right)\right\rangle+\int_0^{t_k^{(n)}}\left\langle f\left(B_u,Y_u\right),\varphi'\left(\cdot+B_u-B_{t_k^{(n)}}\right)\right\rangle du +C^{(n),3}_k\right).
\end{align*}\normalsize
Using (\ref{uvu}) this expression can be simplified to \small
\begin{equation}\label{Znu}
\begin{split}
\left\langle Z^{(n),u}_k,\varphi\right\rangle=\int_0^{t_k^{(n)}}\left(\left\langle f_y\left(B_u,Y_u\right),\varphi\left(\cdot+B_u-B_{t_k^{(n)}}\right)\right\rangle+C^{(n),2}_k(u)\right) Z^{(n),Y}_udu-Z^{(n),B}_k\left(\left\langle u(t), \varphi'\right\rangle+C^{(n),3}_k\right)\quad\\
\qquad+\int_0^{t_k^{(n)}} \left(\left\langle f_b\left(B_u,Y_u\right),\varphi\left(\cdot+B_u-B_t\right)\right\rangle+\left\langle f\left(B_u,Y_u\right),\varphi'\left(\cdot+B_u-B_{t_k^{(n)}}\right)\right\rangle+C^{(n),1}_k(u)\right)Z^{(n),B}_udu+C_k^{(n)}.
\end{split}
\end{equation}\normalsize
 
The next lemma establishes a-priori estimates on $Z^{(n),B}$ and $Z^{(n),Y}$. They will be used afterwards to establish tightness of $\left(Z^{(n),B},Z^{(n),u}\right), n\in\N$.

\begin{lem}\label{bound}
Under Assumptions \ref{initial}, \ref{M}, \ref{Markov}, \ref{scaling}, \ref{h}, \ref{pxx}, \ref{fxx}, and \ref{scaling2}, there exist random variables $C_n,K_n$ with $C_n\ra0$ and $K_n\ra K\in\R_+$ in probability such that for all $n\in\N$ and $k\leq T_n$,
\begin{eqnarray*}
\left|Z^{(n),B}_k\right|+\left|Z^{(n),Y}_k\right|\leq C_n+K_n\left(1+Te^{TK_n}\right)\sup_{k\leq T_n}\left|N_k^{(n)}\right|.
\end{eqnarray*}
\end{lem}

\begin{proof} 
Throughout the proof we denote by $C_n$ a generic positive random variable, not depending on $\varphi$ and converging to zero in probability, which may vary from line to line. 
From the above representation we have for any $\varphi\in H^3$ and $k\leq T_n$ the estimate
\begin{eqnarray*}
\left|\left\langle Z^{(n),u}_k,\varphi\right\rangle\right|
&\leq& \left\Vert\varphi\right\Vert_{H^2}\left[C_n+\left(\sup_{b,y}\left\Vert f_b(b,y)\right\Vert_{L^2}+\sup_{b,y}\left\Vert f(b,y)\right\Vert_{L^2}+C_n\right)\int_0^{t_k^{(n)}} \left|Z^{(n),B}(u)\right|du\right.\\
&&\left.+\left(\sup_{b,y}\left\Vert f_y(b,y)\right\Vert_{L^2}+C_n\right)\int_0^{t_k^{(n)}}\left| Z^{(n),Y}(u)\right|du+\left(\sup_{t\in[0,T]}\left\Vert u(t)\right\Vert_{L^2}+C_n\right)\left|Z^{(n),B}_k\right|\right]\\
&\leq&\left\Vert\varphi\right\Vert_{H^2}\left(C_n+k_n\int_0^{t_k^{(n)}} \left|Z^{(n),B}(u)\right|du+k_n\int_0^{t_k^{(n)}}\left| Z^{(n),Y}(u)\right|du+k_n\left|Z^{(n),B}_k\right|\right),
\end{eqnarray*}
where $k_n$ is converging in probability to some constant $k\in\R_+$. Now choosing $\varphi=h$ we have with $\tilde{K}_n:=\left\Vert\varphi\right\Vert_{H^2}k_n$ for any $k\leq T_n$ the estimate
\[\left|Z^{(n),Y}_k\right|
\leq C_n+\tilde{K}_n\int_0^{t_k^{(n)}} \left|Z^{(n),B}(u)\right|du+\tilde{K}_n\int_0^{t_k^{(n)}}\left| Z^{(n),Y}(u)\right|du+\tilde{K}_n\left|Z^{(n),B}_k\right|.\]
Similarly, Lemma \ref{ZnB} implies that\small
\begin{eqnarray*}
\left|Z^{(n),B}_k\right|\leq C_n+\left|N^{(n)}_k\right|+\left(\sup_{b,y}\left|p_b(b,y)\right|+C_n\right)\int_0^{t^{(n)}_k}\left|Z^{(n),B}(u)\right|du+\left(\sup_{b,y}\left|p_y(b,y)\right|+C_n\right)\int_0^{t^{(n)}_k}\left|Z^{(n),Y}(u)\right|du.
\end{eqnarray*}\normalsize
Applying Lemma \ref{discretization} to the right hand side of the last two inequalities shows that there exists a constant $K\in\R_+$ such that for all $n\in\N$ and $k\leq T_n$,
\begin{eqnarray*}
\left|Z^{(n),B}_k\right|+\left|Z^{(n),Y}_k\right|\leq C_n+\left(C_n+K\right)\left|N_k^{(n)}\right|+\left(C_n+K\right)\Delta t^{(n)}\sum_{j=0}^{k-1}\left(\left|Z^{(n),B}_j\right|+\left|Z^{(n),Y}_j\right|\right).
\end{eqnarray*}
Now Gronwall's lemma implies that there exist $K_n,\ n\in\N,$ converging to $K$ in probability such that for all $n\in\N$ and all $k\leq T_n$,
\begin{eqnarray*}
\left|Z^{(n),B}_k\right|+\left|Z^{(n),Y}_k\right|\leq C_n+K_n\left|N_k^{(n)}\right|+K_ne^{TK_n}\Delta t^{(n)}\sum_{j=0}^{k-1}\left|N_j^{(n)}\right|\leq C_n+K_n\left(1+Te^{TK_n}\right)\sup_{k\leq T_n}\left|N_k^{(n)}\right|.
\end{eqnarray*}
\end{proof}

Below we will use the bound from Lemma \ref{bound} to prove tightness. For this we set
\[D_n:= C_n+K_n\left(1+Te^{TK_n}\right)\sup_{k\leq T_n}\left|N_k^{(n)}\right|\]
and note that, since each $N^{(n)}$ is a martingale, Doob's maximal inequality implies that $(D_n)_{n\in\N}$ is bounded in probability: for any constant $C>0$,
\begin{eqnarray*}
\p\left(\sup_{k\leq T_n}\left|N_k^{(n)}\right|>C\right)\leq C^{-2}\E\left(N^{(n)}_{T_n}\right)^2\stackrel{\p}{\longrightarrow}\frac{1}{C^2}\int_0^T\sigma^2_B(B_t,Y_t)dt,
\end{eqnarray*}
where the convergence follows from the proof of Lemma \ref{N}.

\begin{thm}\label{tightness}
Under Assumptions \ref{initial}, \ref{M}, \ref{Markov}, \ref{scaling}, \ref{h}, \ref{pxx}, \ref{fxx} and \ref{scaling2}, 
the sequence  $\left(Z^{(n),B},Z^{(n),u}\right)_{n\in\N}$ is tight in $D\left([0,T];\R\times H^{-3}\right)$.
\end{thm}

\begin{proof}
By Lemma \ref{discretization} we have
\[\sup_{t\in[0,T]}\left\Vert Z^{(n),u}_{\lfloor t/\Delta t^{(n)}\rfloor}-Z^{(n),u}(t)\right\Vert_{L^2}\ra0\qquad\text{and}\qquad \sup_{t\in[0,T]}\left|Z^{(n),B}_{\lfloor t/\Delta t^{(n)}\rfloor}-Z^{(n),B}(t)\right|\ra0.\]
Therefore, it suffices to prove the tightness of the discrete processes \mbox{$\left(Z^{(n),B}_k,Z^{(n),u}_k\right)_{k\leq T_n},\ n\in\N$,} considered as piecewise constant processes in $D\left([0,T];\R\times H^{-3}\right)$. We will first show tightness of $\left(Z^{(n),B},\langle Z^{(n),u},\varphi\rangle\right)_{n\in\N}$ in $D([0,T];\R\times\R)$ for any $\varphi\in H^3$. This will give tightness of $\left(Z^{(n),B},Z^{(n),u}\right)_{n\in\N}$ in $D([0,T];\R\times\mathcal{E}')$. Afterwards we will show that tightness even holds in $D([0,T];\R\times H^{-3})$.

For $t\geq0$ we define $k_n(t):=\max\{k\in\N:\ k\Delta t^{(n)}\leq t\wedge T\}$. As before we denote by $(C_n)$ a sequence of positive random variables converging to zero in probability, which may vary from line to line. Similarly, $C$ denotes a generic positive constant varying from line to line. To prove tightness we proceed similarly to the proof of Lemma \ref{bound} to get a bound for the fluctuations between time $t+\delta$ and $t$. First we note that Lemmata \ref{ZnB}, \ref{discretization}, and \ref{bound} imply the existence of a constant $K\in\R_+$ such that for all $t\in[0,T)$ and $\delta\in(0,1)$, 
\begin{eqnarray*}
\left|Z_{k_n(t+\delta)}^{(n),B}-Z^{(n),B}_{k_n(t)}\right|&\leq& C_n+
\left|N^{(n)}_{k_n(t+\delta)}-N^{(n)}_{k(t)}\right|+(C+C_n)\Delta t^{(n)}\sum_{j=k_n(t)+1}^{k_n(t+\delta)}\left|Z^{(n),B}_j\right|+\left|Z^{(n),Y}_j\right|\\
&\leq&C_n+
\left|N^{(n)}_{k_n(t+\delta)}-N^{(n)}_{k_n(t)}\right|+(C+C_n)D_n\left(\delta+\Delta t^{(n)}\right).
\end{eqnarray*}
Let $\varphi\in H^3$. To get a similar estimate for the increments of $\langle Z^{(n),u},\varphi\rangle$ we have to be a little bit more careful, because $\langle Z^{(n),u},\varphi\rangle$ is not given in standard semimartingale form in (\ref{Znu}). That is why we need our a priori estimate from Lemma \ref{bound}. We have for all $t\in[0,T)$ and $\delta\in(0,1)$, \small
\begin{align*}
&\left|\left\langle Z^{(n),u}_{k_n(t+\delta)}-Z^{(n),u}_{k_n(t)},\varphi\right\rangle\right|\leq
\left|Z_{k_n(t+\delta)}^{(n),B}\left(\left\langle \varphi',u\left(t_{k_n(t+\delta)}^{(n)}\right)\right\rangle+C^{(n),3}_{k_n(t+\delta)}\right)-Z^{(n),B}_{k_n(t)}\left(\left\langle \varphi',u\left(t_{k_n(t)}^{(n)}\right)\right\rangle+C^{(n),3}_{k_n(t)}\right)\right|\\
&\qquad+\left\Vert\varphi\right\Vert_{H^2}\left(
C_n+k_n\int_{t^{(n)}_{k_n(t)}}^{t_{k_n(t+\delta)}^{(n)}}\left|Z^{(n),B}(u)\right| +\left|Z^{(n),Y}(u)\right|du\right)\\ 
&\qquad+\int_0^{t_{k_n(t)}^{(n)}}\left(\left|\left\langle f_y\left(B_u,Y_u\right),\varphi\left(\cdot+B_u-B_{t_{k_n(t+\delta)}^{(n)}}\right)-\varphi\left(\cdot+B_u-B_{t_{k_n(t)}^{(n)}}\right)\right\rangle\right|+C_n\left\Vert\varphi\right\Vert_{H^2}\right) \left|Z^{(n),Y}(u)\right|du\\ 
&\qquad+\int_0^{t_{k_n(t)}^{(n)}} \left|\left\langle f_b\left(B_u,Y_u\right),\varphi\left(\cdot+B_u-B_{t^{(n)}_{k_n(t+\delta)}}\right)-\varphi\left(\cdot+B_u-B_{t^{(n)}_{k_n(t)}}\right)\right\rangle\right|\left|Z^{(n),B}(u)\right|du\\
&\qquad+\int_0^{t_{k_n(t)}^{(n)}} \left(\left|\left\langle f\left(B_u,Y_u\right),\varphi'\left(\cdot+B_u-B_{t_{k_n(t+\delta)}^{(n)}}\right)-\varphi'\left(\cdot+B_u-B_{t_{k_n(t)}^{(n)}}\right)\right\rangle\right|+C_n\left\Vert\varphi\right\Vert_{H^2}\right)\left|Z^{(n),B}(u)\right|du\\
&\quad\leq D_n\left|\left\langle \varphi',u\left(t_{k_n(t+\delta)}^{(n)}\right)-u\left(t_{k_n(t)}^{(n)}\right)\right\rangle\right|+
D_nC\int_0^{t_{k_n(t)}^{(n)}}\left\Vert \varphi\left(\cdot+B_u-B_{t_{k_n(t+\delta)}^{(n)}}\right)-\varphi\left(\cdot+B_u-B_{t_{k_n(t)}^{(n)}}\right)\right\Vert_{H^1} du\\
&\qquad+\left\Vert\varphi\right\Vert_{H^2}\left[(C+C_n)\left|Z_{k_n(t+\delta)}^{(n),B}-Z^{(n),B}_{k_n(t)}\right|+C_nD_n\left(t^{(n)}_{k_n(t)}+1\right)+C_n +k_nD_n\left(t_{k_n(t+\delta)}^{(n)}-t^{(n)}_{k(t)}\right)\right]\\
&\quad\leq
\left\Vert\varphi\right\Vert_{H^2}\left[
C_n+(C+C_n)D_n\left(\delta +\Delta t^{(n)}\right)+(C+C_n)
\left|N^{(n)}_{k_n(t+\delta)}-N^{(n)}_{k_n(t)}\right|+C_nD_n\left(T+1\right)\right]\\ 
&\qquad+ D_n\left\Vert u\right\Vert_{L^2}\left\Vert \varphi''\right\Vert_{L^2}\left(t_{k_n(t+\delta)}^{(n)}-t^{(n)}_{k_n(t)}\right)+D_nC\left\Vert \int_{B_{t_{k_n(t)}^{(n)}}}^{B_{t_{k_n(t+\delta)}^{(n)}}}\varphi'\left(\cdot-y\right)dy\right\Vert_{H^1} \\
&\quad\leq\left\Vert\varphi\right\Vert_{H^2}\left( o_\p(1)+\mathcal{O}_\p(1)\left(\delta +\Delta t^{(n)}\right)+\mathcal{O}_\p(1)\left|N^{(n)}_{k_n(t+\delta)}-N^{(n)}_{k_n(t)}\right|+\mathcal{O}_\p(1) \left|B_{t_{k_n(t+\delta)}^{(n)}}-B_{t_{k_n(t)}^{(n)}}\right|\right),
\end{align*}\normalsize
where the $o_\p$- and $\mathcal{O}_\p$-terms are uniform in $t$ and $\delta$. Since $B$ is Lipschitz continuous, we get
\begin{eqnarray*}
\left|Z^{(n),B}_{k(t+\delta)}-Z^{(n),B}_{k(t)}\right|+\left|\left\langle Z^{(n),u}_{k(t+\delta)}-Z^{(n),u}_{k(t)},\varphi\right\rangle\right|\qquad\qquad\qquad\qquad\qquad\qquad\qquad\qquad\qquad\\
\qquad\qquad\leq\left(1+\left\Vert\varphi\right\Vert_{H^2}\right)\left(
o_\p(1)+\mathcal{O}_\p(1)\left(\delta +\Delta t^{(n)}\right)+\mathcal{O}_\p(1)\left|N^{(n)}_{k(t+\delta)}-N^{(n)}_{k(t)}\right|\right).
\end{eqnarray*}
Moreover from the proof of Theorem \ref{N} we have for any $\eps>0$,\small
\begin{eqnarray*}
\p\left(\left|N^{(n)}_{k(t+\delta)}-N^{(n)}_{k(t)}\right|>\eps\right)&\leq& \eps^{-2}\cdot\E\left(N^{(n)}_{k(t+\delta)}-N^{(n)}_{k(t)}\right)^2\\
&\leq&\frac{\Delta t^{(n)}}{\eps^2}\sum_{k=k(t)+1}^{k(t+\delta)}p^{(n),A}\left(B^{(n)}_{k-1},Y^{(n)}_{k-1}\right)+p^{(n),B}\left(B^{(n)}_{k-1},Y^{(n)}_{k-1}\right)
\leq \frac{C\left(\delta+\Delta t^{(n)}\right)}{\eps^2}.
\end{eqnarray*}\normalsize
The above computations remain true if we replace $t$ with a bounded stopping time, possibly depending on $n\in\N$, because all error terms are uniform in time. 
Therefore, for any sequence of bounded $\left(\F^{(n)}_k\right)_k$-stopping times $(\tau_n)$ and constants $\delta_n>0$ with $\delta_n\ra0$ we have
\[
\left|Z^{(n),B}_{k_n(\tau_n+\delta_n)}-Z^{(n),B}_{k(\tau_n)}\right|+\left|\left\langle Z^{(n),u}_{k_n(\tau_n+\delta_n)}-Z^{(n),u}_{k(\tau_n)},\varphi\right\rangle\right|\stackrel{\p}{\longrightarrow}0.
\]
Hence, by Aldous' criterion (cf.~Theorem 6.8 in \cite{Walsh}) the sequence $\left(Z^{(n),B},\langle Z^{(n),u},\varphi\rangle\right)_{n\in\N}$ is tight as a sequence in $D\left([0,T];\R\times\R\right)$. Since $\varphi\in H^3\subset \mathcal{E}$ was arbitrary, this already proves tightness of $\left(Z^{(n),B},Z^{(n),u}\right)_{n\in\N}$ in $D\left([0,T];\R\times\mathcal{E'}\right)$ by Mitoma's theorem (cf.~\cite{Walsh}, Theorem 6.13).

Finally, tightness in $D\left([0,T];\R\times H^{-3}\right)$ follows from Corollary 6.16 in \cite{Walsh} and the estimate
\begin{eqnarray*}
\left|\left\langle Z^{(n),u}_k,\varphi\right\rangle\right|
&\leq&\left\Vert\varphi\right\Vert_{H^2}\left(C_n+k_n\int_0^{t_k^{(n)}} \left|Z^{(n),B}(u)\right|du+k_n\int_0^{t_k^{(n)}}\left| Z^{(n),Y}(u)\right|du+k_n\left|Z^{(n),B}_k\right|\right)\\
&\leq&\left(C_n+(T+1)k_nD_n\right)\left\Vert\varphi\right\Vert_{H^2}.
\end{eqnarray*}
\end{proof}

By choosing $\varphi=h\in H^3$, Theorem \ref{tightness} implies that $\left(Z^{(n),B},Z^{(n),Y}\right), n\in\N,$ is tight in $D([0,T];\R\times\R)$. From Lemma \ref{ZnB} and equation (\ref{Znu}) with $\varphi=h$ we deduce that the limit $\left(Z^B,Z^Y\right)$ of any convergent subsequence satisfies
\begin{equation}\label{ZBY}
\begin{split}
Z^B_t=&\int_0^tp_b(B_s,Y_s)Z_s^Bds+\int_0^tp_y(B_s,Y_s)Z_s^Yds+\int_0^t\sigma_B(B_s,Y_s)dW^B_s,\\
Z^Y_t=&\int_0^t\left\langle f_b(B_s,Y_s),h(\cdot+B_s-B_t)\right\rangle Z_s^Bds+\int_0^t\left\langle f_y(B_s,Y_s),h(\cdot+B_s-B_t)\right\rangle Z_s^Yds\\
&+\int_0^t\left\langle f(B_s,Y_s),h'(\cdot+B_s-B_t)\right\rangle Z_s^Bds-Z_t^B\langle u(t),h'\rangle,\quad t\in[0,T].
\end{split}
\end{equation}
Clearly, the system (\ref{ZBY}) has a unique strong solution by the usual Gronwall argument. This gives weak convergence of $\left(Z^{(n),B},Z^{(n),Y}\right)$ to $\left(Z^B,Z^Y\right)$. Especially, the sequence $\left(Z^{(n),B},Z^{(n),Y}\right)_{n\in\N}$ is C-tight, which implies joint tightness of $\left(Z^{(n),B},Z^{(n),Y},Z^{(n),u}\right)$ in $D\left([0,T];\R\times\R\times H^{-3}\right)$. Therefore, (\ref{Znu}) allows us to identify for any $\varphi\in H^3$ the limit of $\langle Z^{(n),u},\varphi\rangle$ uniquely as
\begin{equation}\label{Zu}
\begin{split}
\left\langle Z^u_t,\varphi\right\rangle=&\int_0^t\left\langle f_b(B_s,Y_s),\varphi(\cdot+B_s-B_t)\right\rangle Z_s^Bds+\int_0^t\left\langle f_y(B_s,Y_s),\varphi(\cdot+B_s-B_t)\right\rangle Z_s^Yds\\
&+\int_0^t\left\langle f(B_s,Y_s),\varphi'(\cdot+B_s-B_t)\right\rangle Z_s^Bds-Z_t^B\langle u(t),\varphi'\rangle,\quad t\in[0,T].
\end{split}
\end{equation}

We are now ready to state the second main result of this paper, which describes the limiting dynamics of the joint fluctuations of the price and volume process of the limit order book model around its first order approximation under slow rescaling. The limiting dynamics will be a weak solution to the SPDE (\ref{SPDE2}) as defined in the following theorem. 

\begin{thm}\label{main2}
Under Assumptions \ref{initial}, \ref{M}, \ref{Markov}, \ref{scaling}, \ref{h}, \ref{pxx}, \ref{fxx}, and \ref{scaling2}, we have weak convergence of $Z^{(n)}=\left(Z^{(n),B},Z^{(n),u}\right)$ in $D\left([0,T];\R\times H^{-3}\right)$ to $\left(Z^B,Z^u\right)$, which is the unique solution to the following system: for any $\varphi\in H^3$ we have for all $t\in[0,T|$,\small
\begin{equation}\label{zbu}
\begin{split}
Z^B_t&=\int_0^tp_b(B_s,Y_s)Z^B_sds+\int_0^tp_y(B_s,Y_s)\langle Z^u_s,h\rangle ds+\int_0^t\sigma_B(B_s,Y_s)dW_s,\\
\left\langle Z^u_t,\varphi\right\rangle&=\int_0^t\left\langle\partial_x u(s),\varphi\right\rangle dZ_s^B-\int_0^t\left\langle Z_s^u,\varphi'\right\rangle dB_s+\int_0^t\left\langle f_b(B_s,Y_s),\varphi\right\rangle Z^B_sds+\int_0^t\left\langle f_y(B_s,Y_s),\varphi\right\rangle \langle Z^u_s,h\rangle ds,
\end{split}
\end{equation}\normalsize
where $W$ is a standard Brownian motion and the dynamics of $(B,Y)$ are given in Theorem \ref{LLN}. $(Z^B,Z^u)$ is called a weak solution to (\ref{SPDE2}).
\end{thm}

\begin{proof}
The result follows directly from Theorem \ref{tightness} and equations (\ref{ZBY}) and (\ref{Zu}), once 
we show that the two representations for $Z^u$ in (\ref{Zu}) and (\ref{zbu}) are equivalent. To this end we have 
for any $\varphi\in H^3$, starting from (\ref{Zu}) and applying It\^{o}-Wentzell formula,\small
\begin{eqnarray*}
\left\langle Z^u_t,\varphi\right\rangle&=&\int_0^t\left\langle f_b(B_s,Y_s),\varphi(\cdot+B_s-B_t)\right\rangle Z_s^Bds+\int_0^t\left\langle f_y(B_s,Y_s),\varphi(\cdot+B_s-B_t)\right\rangle Z_s^Yds\\
&&+\int_0^t\left\langle f(B_s,Y_s),\varphi'(\cdot+B_s-B_t)\right\rangle Z_s^Bds-Z_t^B\langle u(t), \varphi'\rangle\\
&=&\int_0^t\left\langle f_b(B_s,Y_s),\varphi\right\rangle Z_s^Bds-\int_0^t\int_0^s\left\langle f_b(B_r,Y_r),\varphi'(\cdot+B_r-B_s)\right\rangle Z_r^BdrdB_s\\
&&+\int_0^t\left\langle f_y(B_s,Y_s),\varphi\right\rangle Z_s^Yds-\int_0^t\int_0^s\left\langle f_y(B_r,Y_r),\varphi'(\cdot+B_r-B_s)\right\rangle Z_r^YdrdB_s\\
&&+\int_0^t\left\langle f(B_s,Y_s),\varphi'\right\rangle Z_s^Bds-\int_0^t\int_0^s\left\langle f(B_r,Y_r),\varphi''(\cdot+B_r-B_s)\right\rangle Z_r^BdrdB_s\\
&&-\int_0^tZ_s^B\langle \partial_tu(s), \varphi'\rangle ds-\int_0^t\langle u(s), \varphi'\rangle dZ_s^B\\
&=&\int_0^t\left\langle f_b(B_s,Y_s),\varphi\right\rangle Z^B_sds+\int_0^t\left\langle f_y(B_s,Y_s),\varphi\right\rangle Z^Y_sds+\int_0^t\left\langle\partial_x u(s),\varphi\right\rangle dZ_s^B\\
&&+\int_0^t\left\langle f(B_s,Y_s),\varphi'\right\rangle Z_s^Bds-\int_0^t\left[\left\langle Z_s^u,\varphi'\right\rangle+Z_s^B\left\langle u(s),\varphi''\right\rangle\right] dB_s-\int_0^tZ_s^B\langle \partial_tu(s), \varphi'\rangle ds\\
&=&\int_0^t\left\langle f_b(B_s,Y_s),\varphi\right\rangle Z^B_sds+\int_0^t\left\langle f_y(B_s,Y_s),\varphi\right\rangle Z^Y_sds+\int_0^t\left\langle\partial_x u(s),\varphi\right\rangle dZ_s^B\\
&&-\int_0^t\left[\left\langle Z_s^u,\varphi'\right\rangle-Z_s^B\left\langle u_x(s),\varphi'\right\rangle\right] dB_s-\int_0^tZ_s^B\left[p^B(B_s,Y_s)-p^A(B_s,Y_s)\right]\langle \partial_xu(s), \varphi'\rangle ds\\
&=&\int_0^t\left\langle f_b(B_s,Y_s),\varphi\right\rangle Z^B_sds+\int_0^t\left\langle f_y(B_s,Y_s),\varphi\right\rangle Z^Y_sds+\int_0^t\left\langle\partial_x u(s),\varphi\right\rangle dZ_s^B-\int_0^t\left\langle Z_s^u,\varphi'\right\rangle dB_s,
\end{eqnarray*}\normalsize
which equals the dynamics in (\ref{zbu}). Since we could go backwards as well in the above equation chain, this shows that (\ref{Zu}) and (\ref{zbu}) are indeed equivalent. Especially, this implies that there exists a unique solution to (\ref{zbu}), because as argued above there exists a unique solution to (\ref{ZBY}) and hence also to (\ref{Zu}). 
\end{proof}

\begin{rem}
The second order approximation derived in this section allows for more frequent price movements, so that we get a non-degenerate first order approximation for both, volumes and prices. The price to pay is that we cannot use the fast rescaling rate anymore, but have to rescale by the slow $\Delta x^{(n)}$-rate corresponding to the time scale of price changes. Since volume changes appear on a faster time scale, this rules out any fluctuations coming from volume placements respectively cancelations (C-events) in the second order approximation. 
\end{rem}

\s{Conclusion}\label{end}

In this paper we have derived two different second order approximations for the multiscale discrete limit order book dynamics studied in \cite{HK1}, corresponding to a fast and a slow rescaling regime. Both rescaling regimes yield two degenerate SDE-SPDE systems.
The fast rescaling can be applied if price movements are supposed to be sufficiently rare and therefore disappear in the first order approximation. In this case, the second order approximation is driven by the fluctuations of the price {\sl and} the fluctuations of order placements and cancelations. However, the SPDE describing the volume fluctuations in the limit degenerates to an infinite dimensional SDE. By contrast, fluctuations from placements and cancelations disappear under the slow rescaling regime. In this case the second order approximation is driven by the fluctuations of the price process only and hence the volume fluctuations are described by a PDE with random coefficients in the limit, which can be seen as a degenerate SPDE. 

Finally, we notice that the critical case that would at least formally correspond to the case $\alpha=1$ and would potentially lead to a non-degenerate SPDE is ruled out due to Assumption \ref{scaling}. Since $\alpha<1$, we are faced with a discrete {\it multiscale} Markov system. In fact, it was already noted in \cite{multiscale} that central limit theorems for multiscale Markov chains can only be derived under very special conditions, even in the finite dimensional case.

\appendix

\s{A technical estimate for $L^2$-martingales}

\begin{lem}[\cite{Pisier}]\label{Pisier}
There exists a constant $C > 0$ such that for all martingale differences $(X_i)$ with values in $L^2(\R)$, we have
\[\E\left(\sup_{i\geq 1}\left\Vert X_i\right\Vert_{L^2}\right)\leq C\E\left(\sum_{i=1}^\infty\left\Vert X_i-X_{i-1}\right\Vert_{L^2}^2\right)^{1/2}.\]
\end{lem}

\iffalse

\s{Mean value theorem for random arguments}

\begin{lem}\label{meanvalue}
Let $f\in C^1$ and $X,Y$ be random variables on $(\Omega,\F,\p)$. Then there exists an $\F$-measurable random variable $Z$ such that
\[f(X)-f(Y)=Z(X-Y).\]
Moreover, if $X_n\ra X$ in probability, then there exists a sequence of random variables $(Z_n)$ such that
\[f(X_n)-f(X)=Z_n(X_n-X)\quad\forall\ n\qquad\text{and}\qquad Z_n\stackrel{\p}{\longrightarrow} f'(X).\]
\end{lem}
 
 \begin{proof}
 The claim follows easily by choosing
 \[Z=\1_{\{X\neq Y\}}\left[\frac{1}{X-Y}\int_Y^Xf'(u)du\right].\]
 \end{proof}
 
 \fi
 
{\small

\bibliographystyle{abbrv}

\bibliography{LOBLit}

\begin{thebibliography}{10}

\bibitem{AbergelJedidi2011}
F.~Abergel and A.~Jedidi.
\newblock A mathematical approach to order book modeling.
\newblock {\em Int. J. Theor. Appl. Finance}, 16(5):1350025, 2013.

\bibitem{Abergel2015}
F.~Abergel and A.~Jedidi.
\newblock Long-time behavior of a {Hawkes} process--based limit order book.
\newblock {\em SIAM J. Financial Math.}, 6(1):1026--1043, 2015.

\bibitem{Alfonsi}
A.~Alfonsi, A.~Fruth, and A.~Schied.
\newblock Optimal execution strategies in limit order books with general shape
  functions.
\newblock {\em Quantitative Finance}, 10(2):143--157, 2010.

\bibitem{Bacry2014}
E.~Bacry and J.~F. Muzy.
\newblock Hawkes model for price and trades high-frequency dynamics.
\newblock {\em Quant. Finance}, 14(7):1147--1166, 2014.

\bibitem{BHQ}
C.~Bayer, U.~Horst, and J.~Qiu.
\newblock A functional limit theorem for limit order books with state dependent
  price dynamics.
\newblock {\em Ann. Appl. Probab.}, 27(5):2753--2806, 2017.

\bibitem{Cont3}
R.~{Cont} and A.~{de Larrard}.
\newblock {Order book dynamics in liquid markets: limit theorems and diffusion
  approximations}.
\newblock ArXiv e-print 1202.6412v1, 2012.

\bibitem{Cont1}
R.~{Cont} and A.~{de Larrard}.
\newblock {Price dynamics in a Markovian limit order market.}
\newblock {\em {SIAM J. Financ. Math.}}, 4(1):1--25, 2013.

\bibitem{Farmer}
J.~Farmer, L.~Gillemot, F.~Lillo, S.~Mike, and A.~Sen.
\newblock What really causes large price changes?
\newblock {\em Quantitative Finance}, 4(4):383--397, 2004.

\bibitem{Gao}
X.~{Gao} and S.~J. {Deng}.
\newblock {Hydrodynamic limit of order book dynamics}.
\newblock {\em Probability in the Engineering and Informational Sciences},
  pages 1–--30, 2016.

\bibitem{Guo}
X.~{Guo}, Z.~{Ruan}, and L.~{Zhu}.
\newblock {Dynamics of order positions and related queues in a limit order
  book}.
\newblock ArXiv e-print 1505.04810v2, 2015.

\bibitem{HK2}
U.~Horst and D.~Kreher.
\newblock {A diffusion approximation for limit order book models}.
\newblock ArXiv e-print 1608.01795v3, 2017.

\bibitem{HK1}
U.~Horst and D.~Kreher.
\newblock A weak law of large numbers for a limit order book model with state
  dependent order dynamics.
\newblock {\em SIAM J. Fin. Math.}, 8:314—--343, 2017.

\bibitem{HP}
U.~Horst and M.~Paulsen.
\newblock A law of large numbers for limit order books.
\newblock {\em Math.~Oper.~Res.}, 42(4):1280--1312, 2017.

\bibitem{HW}
U.~Horst and X.~Wei.
\newblock A scaling limit for limit order books driven by {Hawkes} random
  measures.
\newblock ArXiv e-print 1709.01292, 2017.

\bibitem{Rosenbaum}
W.~{Huang} and M.~Rosenbaum.
\newblock {Ergodicity and diffusivity of Markovian order book models: a general
  framework}.
\newblock {\em SIAM J. Fin. Math.}, 8(1):874--900, 2017.

\bibitem{JS}
J.~Jacod and A.~Shiryaev.
\newblock {\em Limit Theorems for Stochastic Processes}.
\newblock Grundlehren der mathematischen Wissenschaften. Springer Berlin
  Heidelberg, 2nd edition, 2002.

\bibitem{Rosenbaum_Hawkes}
T.~Jaisson and M.~Rosenbaum.
\newblock Limit theorems for nearly unstable {Hawkes} processes.
\newblock {\em Ann. Appl. Probab.}, 25(2):600--631, 2015.

\bibitem{multiscale}
H.-W. Kang, T.~G. Kurtz, and L.~Popovic.
\newblock Central limit theorems and diffusion approximations for multiscale
  markov chain models.
\newblock {\em Ann. Appl. Probab.}, 24(2):721--759, 04 2014.

\bibitem{Marvin}
M.~Keller-Ressel and M.~M\"uller.
\newblock A {Stefan}-type stochastic moving boundary problem.
\newblock {\em Stoch. PDE: Anal. Comp.}, 4(4):746–--790, 2016.

\bibitem{Lakner2}
P.~Lakner, J.~Reed, and F.~Simatos.
\newblock {Scaling limit of a limit order book model via the regenerative
  characterization of L\'{e}vy trees.}
\newblock {\em Stochastic Systems}, 7(2):342--373, 2017.

\bibitem{Pisier}
G.~{Pisier}.
\newblock {Probabilistic methods in the geometry of Banach spaces.}
\newblock {Probability and analysis, Lect. Sess. C.I.M.E., Varenna/Italy 1985,
  Lect. Notes Math. 1206, 167-241}, 1986.

\bibitem{Walsh}
J.~B. Walsh.
\newblock An introduction to stochastic partial differential equations.
\newblock In: {\it {\`E}cole d'{\'E}t{\'e} de Probabilit{\'e}s de Saint Flour
  XIV - 1984}, pages 265--439, Springer Berlin Heidelberg, 1986.

\bibitem{Zheng2014}
B.~Zheng, F.~Roueff, and F.~Abergel.
\newblock Modelling bid and ask prices using constrained {Hawkes} processes:
  Ergodicity and scaling limit.
\newblock {\em SIAM J. Financial Math.}, 5(1):99--136, 2014.

\bibitem{ZhengZ}
Z.~Zheng.
\newblock Stochastic stefan problems: Existence, uniqueness and modeling of
  market limit orders.
\newblock PhD thesis, Graduate College of the University of Illinois at
  Urbana-Champaign, 2012.

\end{thebibliography}
}

\end{document}